\documentclass{article}

\usepackage[english]{babel}
\usepackage[utf8]{inputenc}
\usepackage[T1]{fontenc}

\usepackage{cite}               % citation in sorted order
\usepackage{bbm}
\usepackage[affil-it]{authblk}

\usepackage{dsfont}
\usepackage{braket,mleftright}
\usepackage{upgreek}
\usepackage{etoolbox}

\usepackage{pgfplots, graphicx, tikz}
\usetikzlibrary{arrows}
\usetikzlibrary{decorations.markings}
\usepackage{subcaption}
\usepackage[labelfont=bf]{caption}
\usepackage{algorithm}
\usepackage{algpseudocode}

\usepackage{amsmath, amssymb, amsfonts, amssymb}
\usepackage{physics}

\usepackage[colorinlistoftodos]{todonotes}
\usepackage{xcolor}

\newcommand{\Rn}{\mathds{R}^{d}}
\newcommand{\R}{\mathds{R}}

\newcommand{\inner}[2]{\langle #1,#2 \rangle}
\newcommand{\vect}[1]{\mathbf{#1}}

\usepackage[thmmarks,  thref, amsmath]{ntheorem}

\theoremstyle{plain}
\newtheorem{defn}{Definition}
\newtheorem{lem}{Lemma}
\newtheorem{theorem}{Theorem}
\newtheorem{model}{Model}
\newtheorem{prop}{Proposition}
\newtheorem{cor}{Corollary}

\theorembodyfont{\upshape}

\theoremheaderfont{\itshape}
\theorembodyfont{\upshape}

\newtheorem{example}{Example}
\newtheorem{remark}{Remark}

\theoremstyle{break}
\theoremheaderfont{\scshape\bfseries}
\theorembodyfont{\upshape}
\theoremsymbol{}

\theoremstyle{nonumberbreak}
\theoremheaderfont{\scshape\bfseries}
\theorembodyfont{\upshape}
\theoremsymbol{\rule{0.7em}{0.7em}}%
\theorempostwork{\setcounter{case}{0}\setcounter{ppart}{0}}

\newtheorem{proof}{Proof}

\theoremstyle{nonumberplain}
\theoremheaderfont{\itshape\bfseries}
\theorembodyfont{\upshape}
\theoremsymbol{}

\theoremstyle{empty}
\theoremsymbol{}

\title{Bounded confidence dynamics and graph control: enforcing consensus}
\author{GuanLin Li}
\affil{Georgia Institute of Technology, Department of Physics, Atlanta, GA 30332}
\author{Sebastien Motsch \thanks{Electronic address: \texttt{smotsch@asu.edu}}}
\author{Dylan Weber \thanks{Electronic address: \texttt{djweber3@asu.edu}}}
\affil{Arizona State University,  School of Mathematical and Statistical Sciences, Tempe, AZ 85257-1804, USA}

\begin{document}
\maketitle
\begin{abstract}
  A generic feature of bounded confidence type models is the formation of clusters of agents.  We propose and study a variant of bounded confidence dynamics with the goal of inducing unconditional convergence to a consensus.  The defining feature of these dynamics which we name the \textit{No one left behind dynamics} is the introduction of a local control on the agents which preserves the connectivity of the interaction network.  We rigorously demonstrate that these dynamics result in unconditional convergence to a consensus.  The qualitative nature of our argument prevents us quantifying how fast a consensus emerges, however we present numerical evidence that sharp convergence rates would be challenging to obtain for such dynamics.  Finally, we propose a relaxed version of the control.  The dynamics that result maintain many of the qualitative features of the bounded confidence dynamics yet ultimately still converge to a consensus as the control still maintains connectivity of the interaction network.
\end{abstract}
\tableofcontents

\medskip
\noindent
{\bf Acknowledgments:} The second author wishes to thank Benedetto Picolli for helpful discussions.

\section{Introduction}
Mathematical models of opinion formation have long been objects of theoretical interest.  Models in this context are often posed in an agent based framework where the potential for agents to interact is encoded in a network \cite{olfati-saber_consensus_2007, xia_opinion_2011, lorenz_continuous_2007, castellano_statistical_2009, hegselmann_opinion_2002, vicsek_novel_1995}.  The rise of social networks as some of the main forums for the exchange of ideas clearly motivates the need to continue the study of these models.  Analysis of the metadata associated with social networks shows that an emergent feature is the formation of polarized communities or "echo chambers" \cite{goldie_using_2014, gilbert_blogs_2009, garimella_political_2018, quattrociocchi_echo_2016}.  In this paper we study a class of models that exhibit this phenomenon.  We are especially interested in the emergence of a \textit{consensus} - when the opinions of all agents agree.

The defining feature of the agent-based approach is the study of how locally defined interaction rules affect globally observed behavior among the agents.  Models of opinion formation often have the feature that agents can only interact if they are connected in an underlying network structure.  Therefore, a hallmark of the study of these models is examining how the interplay between the topology of the underlying network and the interaction rules affect the distribution of opinions among the agents \cite{motsch_heterophilious_2014, saber_consensus_2003, weber_deterministic_2019}.  Of particular interest is how these factors can lead to the emergence of a consensus among the agents.  Models in this context have the generic assumption that the opinion of a given agent is continuously influenced by those to whom it is connected in the network according to relatively simple interaction rules which are globally defined. Often these rules carry an assumption of \textit{local consensus}; if agents interact only with each other then they should agree in some sense.  This assumption can also be interpreted as saying that there are only attractive forces present among the agents.  One might assume that the attractive nature of the interactions causes the emergence of consensus to be a ubiquitous feature of this class of models, however this is not the case. The manner in which agents are connected in the underlying network has a large effect on the distribution of opinions observed among the agents \cite{motsch_heterophilious_2014,  weber_deterministic_2019}.  The interplay between the network structure and the interaction rules can often cause the analysis of these models to be very involved; a popular strategy is to use simplifying assumptions on the network structure such as symmetry of connections or static connections that do not change throughout the evolution of the model \cite{saber_consensus_2003, spanos_dynamic_2005, yu_second-order_2010, olfati-saber_consensus_2007}.  Given an interaction rule, the second strategy could be viewed as studying a "linearization" of a model with the same interaction rule but dynamic connections.  A main takeaway from the study of these models is that a necessary condition for the emergence of a consensus is the persistence of a suitable degree of connectivity in the network throughout the evolution of the dynamics. This allows for "heterophilic" interactions; agents with disparate opinions interact and due to the attractive nature of the interaction rules eventually agree \cite{motsch_heterophilious_2014}.

We study a class of models inspired by the Hegselmann-Krause bounded confidence model \cite{hegselmann_opinion_2002} in which the connections between agents are dynamic; a connection forms between agents when their opinions are within an interaction range.  This dynamic (combined with an attractive interaction rule) causes the formation of "clusters" of opinions in the long time limit to be a generic behavior; consensus is rare.  In fact, it is relatively easy to show that a consensus can only occur if the initial opinions of all the agents are within the interaction range of each other.  For this reason much of the study of this class of models has focused on characterizing the clustering behavior \cite{blondel_krauses_2009, deffuant_mixing_2000, jabin_clustering_2014, lorenz_consensus_2006, blondel_2r_2007, krause_discrete_nodate}.  We take a different viewpoint in this manuscript and instead investigate controls on collections of agents \cite{caponigro2013sparse,piccoli2019sparse}   otherwise evolving according to bounded confidence dynamics that result in consensus.  The interaction range in bounded confidence dynamics causes interactions between agents to be "homophilic"; agents only interact with agents who are sufficiently "similar".  This tendency causes the interaction network of a collection of agents to quickly become disconnected and prevents a consensus from occurring despite the fact that agents who do interact attract each other. Therefore, the controls that we impose on agents are generally motivated by maintaining connectivity among the agents.

We investigate two different ways of augmenting the bounded confidence dynamics with the goal of achieving consensus.  The first strategy which we dub the \textit{no one left behind dynamics} imposes the rule that once agents become connected they remain connected.  We prove rigorously that under this augmentation of the bounded confidence dynamics, connectivity of the initial interaction network is sufficient for a consensus to emerge for agents whose opinion can be of arbitrary dimension.  If we restrict agent opinions to being one dimensional we can quantify how fast a consensus is reached as we derive explicit convergence rates.  Here, an interesting phenomenon is observed as we find that the convergence occurs in two stages.  Before all agents are within the interaction range provided by the bounded confidence dynamics the convergence is linear, afterwards the convergence to consensus spontaneously becomes exponential.  The preservation of connections among agents is sufficient to preserve the connectivity of the network however it isn't necessary.  If the existence of \textit{paths} between agents is maintained then the connectivity of the interaction network is maintained as well.  The second strategy which we dub the \textit{relaxed no one left behind dynamics} takes advantage of this observation and demands that agents who are connected by a \textit{path} in the interaction network remain connected by a path.  We find numerical evidence that this less restrictive control is sufficient for consensus as well. We also demonstrate numerically that this strategy is, in a sense, an interpolation between the bounded confidence dynamics and the no one left behind dynamics - most agents evolve according to the bounded confidence dynamics and a high degree of clustering initially occurs. However several "bridging" agents alter their trajectories in order to maintain connectivity of the interaction network and ensure convergence to a consensus.

% Talk about numerical evidence for slower convergence rate? Two stage convergence?

% how to extend study?

\section{Bounded confidence opinion dynamics}

We will consider a collection of $N$ agents where the opinion of the $i$th agent is denoted by ${\bf x}_{i} \in \Rn$.  We will be concerned with a class of opinion dynamics that take the form:
\begin{equation}
  \label{eq:opinion_formation}
  \dot{\bf x}_{i}= \sum_{j=1}^{N} a_{ij}({\bf x}_{j}-{\bf x}_{i}),\quad  a_{ij}=\frac{\Phi_{ij}}{\sum_{k=1}^N \Phi_{ik}} \;\;\;\text{ with }\; \Phi_{ij}=\Phi(\left | {\bf x}_{j}-{\bf x}_{i} \right |).
\end{equation}
Here, $\Phi$ represents the so-called \textit{interaction function} and can be thought of as encoding how much influence one agent exerts over another, i.e. if $a_{ij} \neq 0$ then agent $j$ is influencing agent $i$ with strength $a_{ij}$. The coefficients $a_{ij}$ can be thought of as encoding the structure of a directed network on which the agents interact. As the $a_{ij}$'s are time dependent, the structure of the network changes in time as well.  We will refer to this network as the \textit{interaction network} and denote it by $G = (V,E)$ where $V$ is the set of agents. Throughout the following we will assume that the interaction function has compact support on the interval $[0,1]$, that it is positive on its support and has a minimum and maximum on this interval:
\begin{equation}
  \label{eq:m_and_M}
  m=\min_{r\in[0,1]}\Phi(r) \quad , \quad M=\max_{r\in[0,1]}\Phi(r).
\end{equation}
This assumption encodes that individuals will only interact and share ideas if their opinions are close enough to begin with. Notice that these conditions on the interaction function allow for a discontinuity at $x = 1$ from above in general; a prototypical example would be the indicator function on the interval $[0,1]$, i.e. $\Phi(r)=\mathbbm{1}_{[0,1]}(r)$.
% \begin{displaymath}
%   \Phi(|{\bf x}_{i}-{\bf x}_{j}|) = \left\{
%       \begin{array}{ll}
%         0 & |{\bf x}_{i}-{\bf x}_{j}| > 1 \\
%         1 & |{\bf x}_{i}-{\bf x}_{j}| \leq 1
%       \end{array}
%     \right.
% \end{displaymath}

This model represents a continuous version of the \textit{bounded confidence} opinion dynamics introduced in \cite{hegselmann_opinion_2002}. Notice that in general $a_{ij} \neq a_{ji}$ and thus the dynamics are not symmetric (the center of mass is not preserved). However, if $a_{ij} > 0$ then we must also have $a_{ji} > 0$, in other words if $j$ influences $i$ then $i$ must influence $j$.
\begin{remark}
  Notice that since the interaction coefficients $a_{ij}$ satisfy $\sum_{j=1}^Na_{ij} = 1$ we can rewrite the dynamics  \eqref{eq:opinion_formation} as:
  \begin{equation}
    \label{eq:opinion_formation_bar}
    \dot{\bf x}_{i}= \overline{\bf x}_{i}-{\bf x}_{i},\quad  \overline{\bf x}_{i}= \sum_{j=1}^{N} a_{ij}{\bf x}_{j},
  \end{equation}
  so ${\bf x}_{i}$ moves towards $\overline{\bf x}_{i}$; the average opinion of all agents within the interaction radius of agent $i$ weighted by their influence on agent $i$ (see figure \ref{fig:bounded_confidence}).
  \begin{figure}[ht]
    \centering
    \includegraphics[scale = 2]{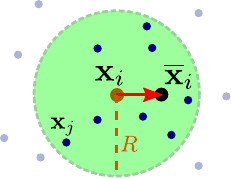}
    \caption{The movement of an agent according to the bounded confidence dynamics \eqref{eq:opinion_formation_bar}.}
    \label{fig:bounded_confidence}
  \end{figure}
\end{remark}
In this manuscript we will be concerned with conditions that cause the opinions of all agents to converge to a consensus.
\begin{defn}
  We say that the dynamics \eqref{eq:opinion_formation} converge to a \textbf{consensus} if there exists ${\bf x}_*$ such that:
  \begin{equation}\label{consensus}
    \lim_{t\rightarrow+\infty}{\bf x}_{i}(t) = {\bf x}_* \quad \text{for any } i.
  \end{equation}
\end{defn}
We will see that due to the local consensus assumption, the notion of \textit{connectivity} is crucial to the formation of a consensus.

\begin{defn}
  \label{connectivity - network}
  We say that the configuration of agents $\set{{\bf x}_{1},...,{\bf x}_{N}}$ is \textbf{connected} if for any two agents $i$ and $j$ there exists a \textit{path} between $i$ and $j$; that is a subset $\set{i_{1},...,i_{m}}\subseteq\set{1,...,N}$ such that $i_{1} = i$, $i_{m} = j$ and:
  \begin{equation*}
    a_{i_{k},i_{k+1}}\neq0\quad\text{for all $k$}.
  \end{equation*}
\end{defn}

Clearly, connectedness of the initial configuration is a necessary condition for the emergence of a consensus.  Due to the local consensus assumption, two agents who initially do not have a path between them will never become connected and therefore will not converge on the same opinion.  However, connectedness of the initial condition is not sufficient for the emergence of a consensus as the dynamics do not necessarily preserve connectedness between two agents (see for instance figure \ref{fig:example_withWithout_ctrl}-left); the dynamics \eqref{eq:opinion_formation} must be modified in some manner for connectedness of the initial condition to be sufficient for consensus.

\section{No one left behind - enforcing consensus}

\subsection{Critical region}

We now modify the dynamics \eqref{eq:opinion_formation} with the aim of preserving connectivity between agents.  We will distinguish between the 1-dimensional and $d$-dimensional cases.  The intuitive idea is to introduce a control on the bounded confidence dynamics that causes an agent to alter its trajectory if it is close to disconnecting with a neighbor.  With this aim in mind we introduce the notion of a \textit{critical region} associated with each agent (see also an illustration in figure \ref{fig:critical_region}).

\begin{defn}
  Fix $0\leq r_* \leq 1$ and let $\set{{\bf x}_{1},...,{\bf x}_{N}}\subseteq \Rn$ be a configuration of agents.  The \textbf{critical region} associated with agent $i$ is given by:
  \begin{equation}
    \label{eq:B_i}
    \mathcal{B}_i = \{{\bf x}\in\mathbb{R}^{d} \;|\; 1-r_*\leq|{\bf x}-{\bf x}_{i}|\leq1 \text{ and } \langle\overline{\bf x}_{i}-{\bf x}_{i},{\bf x}-{\bf x}_{i}\rangle \leq0\}.
  \end{equation}
  where $\overline{\bf x}_i$ is given by \eqref{eq:opinion_formation_bar}.
\end{defn}

\begin{figure}[ht]
  \centering
  \includegraphics[width=.97\textwidth]{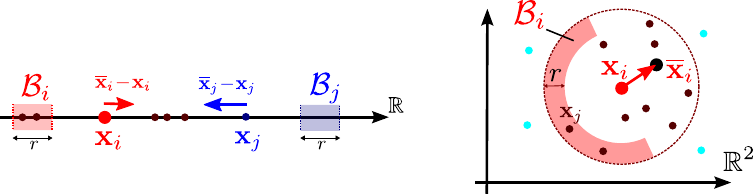}
  \caption{Illustration of the critical regions \eqref{eq:B_i} in $\mathbb{R}$ (interval {\it behind} ${\bf x}_i$) and $\mathbb{R}^2$ (semi-annulus region). The opinion ${\bf x}_i$ is attracted toward the local average $\overline{\bf x}_i$, hence moves with velocity $\overline{\bf x}_i-{\bf x}_i$. In the ``No-left behind dynamics'' \eqref{model:1D_NOLB}, ${\bf x}_i$ can only move only if there is no one in its critical region $\mathcal{B}_i$. Thus, ${\bf x}_i$ freezes whereas ${\bf x}_j$ is free to move in the left illustration.}
  \label{fig:critical_region}
\end{figure}

\subsection{No one left behind dynamics in $\mathbb{R}$}

Notice that the critical region of any agent depends on the local average that the agent will move its opinion towards, \eqref{eq:opinion_formation_bar}; the critical region of an agent is always "behind" the agent in the sense that it is always in the opposite direction of the direction of movement of the agent.  The critical region is the main tool used to enforce connectivity preservation in the bounded confidence dynamics \eqref{eq:opinion_formation}.  We first illustrate the main idea in one dimension.

\begin{model}[1D NOLB]\label{model:1D_NOLB}
  Consider a collection of agents with opinions $\set{x_{1},...,x_{N}}$ in $\R$.  The \textbf{1-D No one left behind} dynamics are given by:
  \begin{equation}
    \label{eq:1D_NOLB}
    \dot{x}_{i}= \mu_i\big(\overline{x}_{i}-{x}_{i}\big),\quad  \mu_i = \left\{
      \begin{array}{ll}
        0 & \text{if there exists } {x}_{j} \in \mathcal{B}_i \\
        1 & \text{otherwise}
      \end{array}
    \right.
  \end{equation}
  where $\overline{x}_{i}$ is the local average defined in \eqref{eq:opinion_formation_bar}.
\end{model}

The scalar $\mu_{i}$ can be interpreted as a local control on the opinion of the $i$th agent.  Under the dynamics given by Model \ref{model:1D_NOLB} an agent evolves according to the normal bounded confidence model \eqref{eq:opinion_formation} unless there is another agent in its critical region in which case it does not move.  In other words if the normal bounded confidence dynamics will cause an agent's opinion to change in such a way that it will become disconnected from one of its neighbors then it will stop moving and not leave its neighbor behind.

\begin{figure}[ht]
  \centering
  \includegraphics[scale=.85]{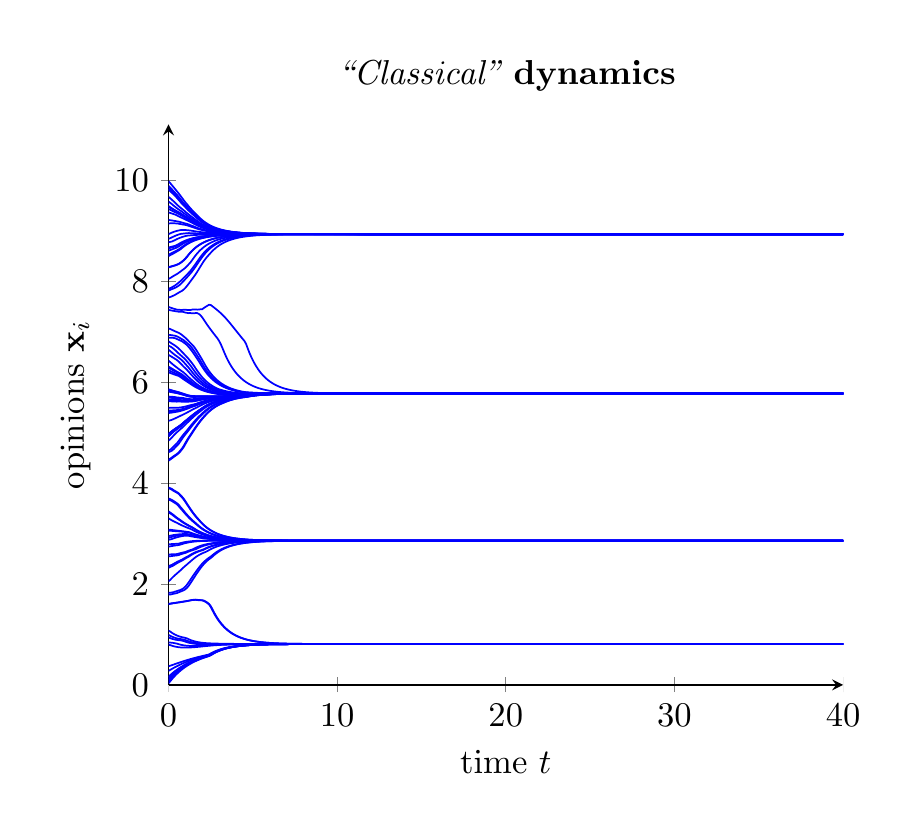} \quad
  \includegraphics[scale=.85]{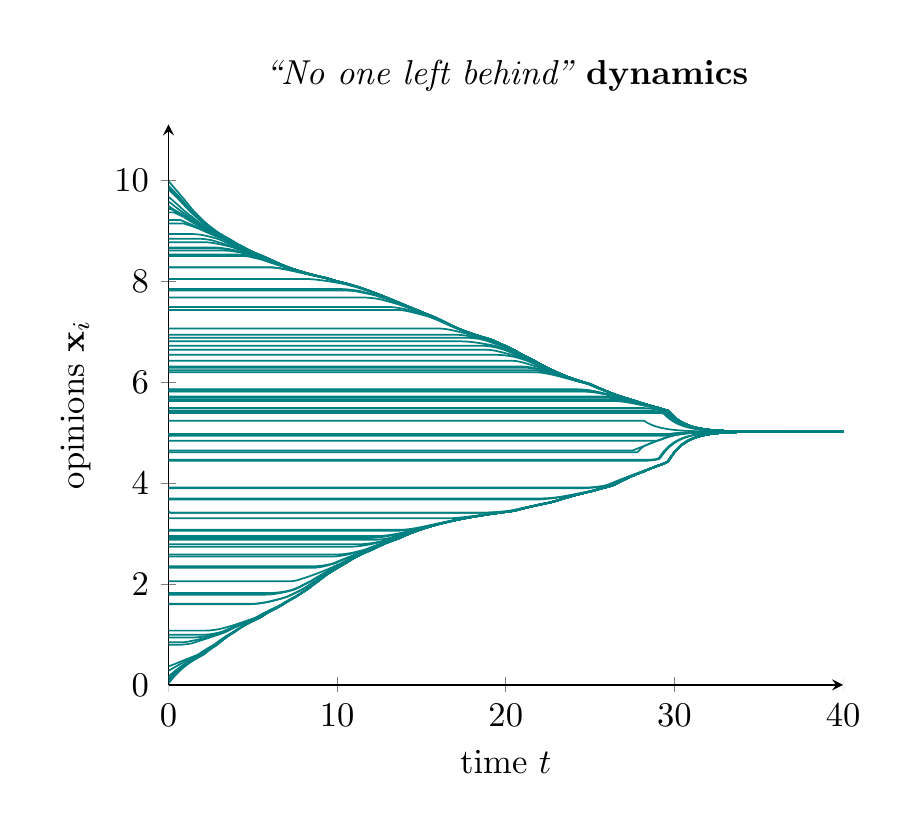} \\
  \caption{Simulation of the opinion dynamics without and with control (resp. left and right figure), e.g. solving resp. \eqref{eq:opinion_formation_bar} and Model \ref{model:1D_NOLB} with $r_{*} = \frac12$. With the control (right), the dynamics converge to a consensus.
    % We also observe that the  diameter $d(t)$ converges exponentially fast once it reaches $d=1$.
  }
  \label{fig:example_withWithout_ctrl}
\end{figure}

\begin{remark}
  \label{rem:behind_graph}
  In addition to the network defined by the interactions among agents, the critical region allows one to impose an additional network structure on the collection of agents; there exists a link from agent $i$ to agent $j$ if agent $j$ is in the critical region of agent $i$.  We will refer to this network as the \textit{behind graph} and denote it by $G^{\mathcal{B}} = (V, E^{\mathcal{B}})$ where $V$ is the set of agents $\set{1\dots N}$ and $E^{\mathcal{B}}$ is the set of edges:
  \begin{equation}
    \label{eq:behind_graph}
    (i,j) \in E^{\mathcal{B}} \text{ if } j \in \mathcal{B}_i.
  \end{equation}
    In this notation we could write \eqref{model:1D_NOLB} as:
  \begin{equation*}
    \dot{x}_{i}= \mu_i\big(\overline{x}_{i}-x_{i}\big),\quad  \mu_i = \left\{
      \begin{array}{ll}
        0 & \text{if} \; (i,j)\in E_{\mathcal{B}} \\
        1 & \text{otherwise}
      \end{array}
    \right.
  \end{equation*}

Note that while the nature of bounded confidence dynamics forces the interaction network to be undirected, the behind graph must be directed as the presence of agent $j$ in the behind region of agent $i$ does \textit{not} imply the opposite.  Note also that if one denotes $G=(V,E)$ as the interaction graph (i.e. $(i,j)\in E$ if $a_{ij}, a_{ji} > 0$), then the behind graph $E^{\mathcal{B}}$ is a directed subgraph of $E$ - see Figure \ref{fig:behind_graph}.

\begin{figure}[ht]
  \centering
  \includegraphics[scale=2.2]{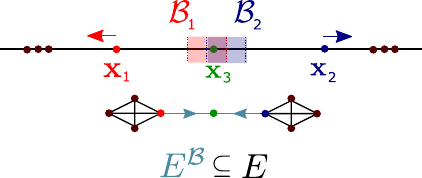}
  \caption{A configuration of agents (top) and the resulting interaction graph (edge set E, black) and behind graph (edge set $E^{\mathcal{B})}$, light blue).  Note that the behind graph is a directed subgraph of the interaction graph.}
    % We also observe that the  diameter $d(t)$ converges exponentially fast once it reaches $d=1$.
  \label{fig:behind_graph}
\end{figure}

\end{remark}

We examine the effect that the augmented dynamics have on the long time behavior of the opinions in Figure \ref{fig:example_withWithout_ctrl} and note that for the same initial condition that normal bounded confidence dynamics do not result in a consensus but four "clusters" of opinions whereas the controlled dynamics preserve the connectivity of the agents and result in a consensus.  Interestingly, we find that the dynamics introduced in Model \eqref{model:1D_NOLB} are not sufficient to ensure consensus in dimensions larger than one.

\subsection{No one left behind dynamics in $\mathbb{R}^d$}

Given that the dynamics defined in Model \eqref{model:1D_NOLB} are sufficient for convergence to a consensus for a connected initial configuration in one dimension (see Remark \ref{rem:1D_special_case} and Theorem 2), one might expect that they should be sufficient for consensus in the $d$ - dimensional case as well.  Interestingly, this is not the case as the following example illustrates.
\begin{example}
  \label{rem:c_ex_2d}
  We illustrate that the dynamics introduced in Model \eqref{model:1D_NOLB} do not ensure consensus in dimension 2 - this example can clearly be generalized to larger dimensions. Consider $6$ {\it clusters} of opinions located on a regular hexagon with equal sides of length $d = 1-\frac{r_*}{2}$ with $r_*\ll 1$. We denote the vertices of this regular hexagon as ${\bf x}_{i}$ with $i=1,\dots,6$, and the number of opinions in the cluster ${\bf x}_{i}$ as $\mathcal{N}({\bf x}_{i})$. Consider for instance, we have the following distribution of agents (see fig.~\ref{fig:counter_ex2D}):
  \begin{displaymath}
    \mathcal{N}({\bf x}_{1}) = \mathcal{N}({\bf x}_6) = 1, \quad \mathcal{N}({\bf x}_{2}) = \mathcal{N}({\bf x}_{5})= 10, \quad \mathcal{N}({\bf x}_{3})=\mathcal{N}({\bf x}_4) = 100.
  \end{displaymath}
  In this setting, all the agents have another agent in their critical region. Thus, if one uses the same dynamics as in the one dimensional setting \eqref{model:1D_NOLB}, we find $\mu_i=0$ for all $i$ and therefore the agents are stuck in this initial configuration. Thus, the "naive" control fails to achieve consensus.

  \begin{figure}[ht]
    \centering
    \includegraphics[scale=0.9]{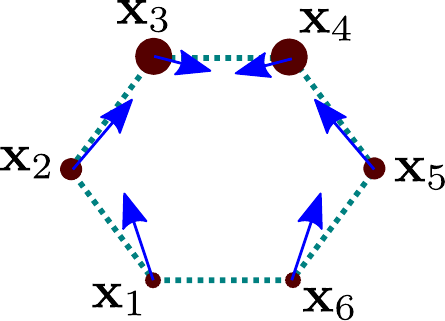}
    \caption{Counter-example in multi-dimension. Blue arrow is the velocity of each cluster. In this setting, every agent has someone in its critical region $\mathcal{B}_i$. Thus, the {\it naive} control in Model \ref{model:1D_NOLB} would prevent anyone from moving.}
    \label{fig:counter_ex2D}
  \end{figure}
\end{example}

Clearly, in the $d$-dimensional case, the condition that an agent may not move if another agent is in its critical region must be weakened in order to achieve a consensus.  However, we still want to maintain the property that an agent may not move away (and possibly disconnect) from agents in its critical region as we know that connectivity is necessary for consensus.  With this in mind we introduce the notion of \textit{admissible velocity}.

\begin{defn}
  Let $\set{{\bf x}_{1},...,{\bf x}_{N}}\subseteq \Rn$ be a configuration of agents.  The \textbf{cone of admissible velocity } associated with agent $i$ is given by:
  \begin{equation}
    \label{eq:Admissible}
    \mathcal{C}_{i} = \{ {\bf v}\in\mathbb{R}^d\,|\;\langle {\bf v}, {\bf x}_{j}-{\bf x}_{i} \rangle \geq 0\quad \text{for all }  {\bf x}_{j}\in \mathcal{B}_i \}.
  \end{equation}
  where $\mathcal{B}_{i}$ is the critical region \eqref{eq:B_i} associated to agent $i$. If the critical region $\mathcal{B}_i$ is empty, then $\mathcal{C}_{i}=\mathbb{R}^d$.
\end{defn}

\begin{remark}
  We note that we can define the cone of admissible velocities in terms of the behind graph introduced in Remark \ref{rem:behind_graph}. We just have to replace ${\bf x}_j\in \mathcal{B}_i$ by $(i,j)\in E^{\mathcal{B}}$ in the definition of $\mathcal{C}_i$.
\end{remark}

\begin{figure}[ht]
  \centering
  \includegraphics[scale=2.5]{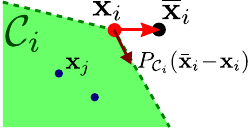}
  \caption{The velocity of agent $i$ is the projection of the {\it desired} velocity $\overline{\bf x}_i-{\bf x}_i$ onto the cone of admissible velocity $\mathcal{C}_{i}$.}
  \label{fig:krause_control_2D}
\end{figure}

We can now weaken the dynamics introduced in Model \ref{model:1D_NOLB} by merely enforcing that the velocity of an agent belong to its cone of admissible velocity via a projection operator \cite{hiriart2012fundamentals}.  Intuitively, instead of forcing an agent to stop whenever its critical region is nonempty it can "take care" of agents in its critical region by moving closer to those agents \textit{and} its local average (if possible).

\begin{model}[NOLB]\label{model:no_one_left_behind}
  Let $\set{{\bf x}_{1},...,{\bf x}_{N}}\subseteq \Rn$ be a configuration of agents.  The \textbf{no one left behind (NOLB)} dynamics are given by:
  \begin{equation}\label{eq:no_one_left_behind}
    {\bf x}_{i}' =  P_{\mathcal{C}_i}\big(\overline{\bf x}_{i} - {\bf x}_{i}\big)
  \end{equation}
  where $\overline{\bf x}_{i}$ is the average velocity defined in \eqref{eq:opinion_formation_bar} and $P_{\mathcal{C}_{i}}: \Rn \rightarrow \mathcal{C}_{i}$ is the projection operator associated to the cone of admissible velocities $\mathcal{C}_i$ \eqref{eq:Admissible}.
\end{model}

% \begin{remark}
%   \label{rem:lagrangian}
%   Since the admissible set $\mathcal{C}_{i}$ is a closed cone, the projection $P_{\mathcal{C}_i}({\bf v})$ is unique and for any ${\bf v}$ satisfies:
%   %   $\mathcal{C}_i$ is a polyhedral
%   \begin{equation}
%     \label{eq:lagrange_mult_proj}
%     {\bf v}-P_{\mathcal{C}_i}({\bf v}) + \sum_{j\in\mathcal{B}_i} \lambda_j ({\bf x}_j-{\bf x}_i)=0 \qquad \text{with } \lambda_j\geq0,
%   \end{equation}
%   where the $\{\lambda_j\}_j$ correspond to the Lagrange multipliers of the minimization problem.
% \end{remark}

\begin{remark}
  \label{rem:1D_special_case}
  We note that the 1-D NOLB dynamics introduced in \eqref{model:1D_NOLB} are a special case of the general NOLB dynamics introduced in Model \ref{model:no_one_left_behind}.  Indeed, in one dimension the cone of admissible velocity is given by
  \begin{equation*}
    \mathcal{C}_{i} = \{{\bf v}\in\mathbb{R}\,|\;(  {\bf v}, {\bf x}_{j}-{\bf x}_{i} )\geq 0\quad \forall  {\bf x}_{j}\in \mathcal{B}_i \}.
  \end{equation*}
  Now, if $\mathcal{B}_{i}$ is empty then we must have that $\mathcal{C}_{i} = \R$ and therefore the projection operator $P_{\mathcal{C}_{i}}$ must be the identity, i.e
  \begin{equation*}
    {\bf x}_{i}' = P_{\mathcal{C}_{i}}\big(\overline{\bf x}_{i} - {\bf x}_{i}\big) = \overline{\bf x}_{i} - {\bf x}_{i}.
  \end{equation*}
  On the other hand, if there exists $x_{j} \in \mathcal{B}_{i}$  (and assuming without loss of generality that $\overline{\bf x}_{i}\geq {\bf x}_{i}$) we must have that ${\bf x}_{j}\leq {\bf x}_{i}$ which implies that:
  \begin{equation*}
    \mathcal{C}_{i} = \set{{\bf v}\in \R | {\bf v}\leq 0}.
  \end{equation*}
  Since $\overline{\bf x}_{i} - {\bf x}_{i}\geq 0$ we therefore must have that
  \begin{equation*}
    {\bf x}_{i}' = P_{\mathcal{C}_{i}}\big(\overline{\bf x}_{i} - {\bf x}_{i}\big) = 0.
  \end{equation*}
\end{remark}

We illustrate the dynamics \eqref{eq:no_one_left_behind} and its long term behavior  in Figure \ref{fig:2D_timelapse}. As in the 1D case, a consensus is reached after some time whereas the classical dynamics generate multiple clusters.

\begin{figure}[ht]
  \centering
  \includegraphics[width=.97\textwidth]{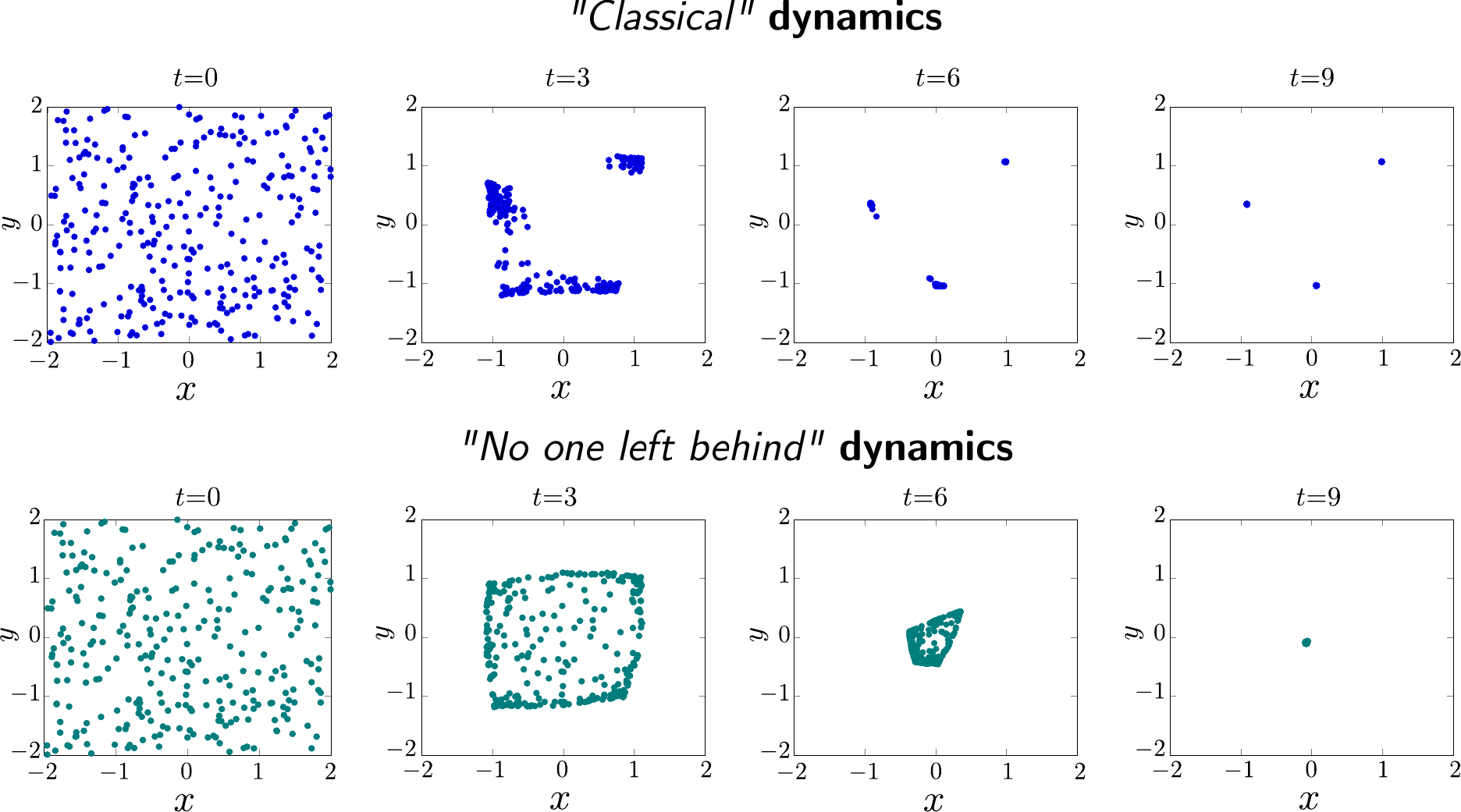}
  \caption{2D simulation of opinion dynamics without and with control (resp, top and bottom figure), e.g. solving resp. \eqref{model:1D_NOLB} and \eqref{eq:no_one_left_behind} with $r_* = \frac12$.  With the control (bottom), the dynamics converge to a consensus.}
  \label{fig:2D_timelapse}
\end{figure}

We now rigorously show that augmenting the bounded confidence dynamics in this manner is sufficient to ensure consensus in the case that the initial configuration of the agents is connected.

\section{Convergence to a consensus}

\subsection{Preservation of connectivity}

A unifying feature of the classical bounded confidence dynamics and the NOLB dynamics is that due to the local consensus assumption they both result in a configuration of agents $\set{{\bf x}_{i}(t)}$ \textit{contracting} in space.  More specifically let us denote the convex hull of a configuration $\set{{\bf x}_{i}(t)}$  by $\Omega(t)$, i.e.
\begin{equation*}
  \Omega(t) = \text{Conv}\{{\bf x}_{1}(t),...,{\bf x}_{N}(t)\}.
\end{equation*}
The agents contract in the following sense (see \cite{jabin_clustering_2014} for a proof).
\begin{prop}\label{prop:contract}
  If $\set{x_{i}}_{i}$ evolves according to the bounded confidence dynamics or the NOLB dynamics then the convex hull $\Omega$ satisfies:
  \begin{equation}
    \Omega(t_{1}) \subset \Omega(t_{0})\quad\text{for any}\quad t_{1}\geq t_{0}
  \end{equation}
\end{prop}
The contractive nature of the dynamics implies that for at least a subsequence of times the configuration approaches a limiting configuration thanks to Bolzano-Weierstrass theorem.
\begin{cor}\label{cor:sequence}
  There exists a limiting configuration $\set{{\bf x}_{i}^{\infty}}_{i}$ and a sequence of times $(t_{n})_{n}$ such that $t_{n}\rightarrow\infty$ and ${\bf x}_{i}(t_{n})\rightarrow {\bf x}_{i}^{\infty}$ as $n\rightarrow \infty$. %Additionally there exists a set $\Omega^{\infty}$ such that $\Omega(t_{n})\rightarrow \Omega^{\infty}$ as $n\rightarrow\infty$  in the sense of sets and $\text{Conv}(\vect{x}^{\infty}) = \Omega^{\infty}$.
\end{cor}
Another consequence is that the so-called diameter of the configuration $d(t)$ is decaying.

\begin{cor}\label{cor:decay_diameter}
  Denote the diameter $d(t) = \max_{1\leq i,j\leq N} |{\bf x}_{i}(t) - {\bf x}_{j}(t)|$. The diameter is non-increasing:
  \begin{equation*}
    d(t_{2})\leq d(t_{1})\quad\text{for any $t_{2}\geq t_{1}$}.
  \end{equation*}
\end{cor}

% \begin{proof} Notice that by Proposition \ref{prop:contract} that we have for any $t$ that
%   \begin{equation*}
%     \set{x_{i}(t)}_{i} \subset \Omega(0).
%   \end{equation*}
%   Therefore Bolzano-Weierstrass theorem implies the first result.  The second assertion follows immediately from the fact that $\Omega(t_{n})$ is a decreasing sequence of sets by Proposition \ref{prop:contract}.
% \end{proof}
So in both the bounded confidence model and the NOLB model, agents remain close to each other in the sense of Proposition \ref{prop:contract}.  However, we have already seen that the bounded confidence dynamics do not converge to a consensus from a connected initial condition; contractiveness alone is not sufficient for consensus.

The fundamental property of the NOLB dynamics \eqref{eq:no_one_left_behind} that distinguish them from the classical bounded confidence dynamics is that they rule out the possibility of agents who are connected later becoming disconnected as illustrated in Figures \ref{fig:example_withWithout_ctrl} and \ref{fig:2D_timelapse}.  The mechanism through which the local control accomplishes this can be seen in the following result.

\begin{prop} \label{prop:distance_decaying}
  Fix $0\leq r_* \leq 1$. Suppose that at time $t$ the opinions of agents $i$ and $j$ evolve according to \eqref{eq:no_one_left_behind} and satisfy: $1 - r_* \leq |{\bf x}_{i}(t) - {\bf x}_{j}(t)|\leq 1$ (i.e. agent $i$ and agent $j$ are within the critical distance of each other), then:
  \begin{equation}
    \label{eq:distance_decaying}
    \frac{d}{dt}|{\bf x}_{i}(t)-{\bf x}_{j}(t)| \leq 0.
  \end{equation}
\end{prop}
\begin{proof}
  As agent $i$ and agent $j$ are within the critical distance of each other we can assume that either agent $i$ is in the critical region of agent $j$ or vice versa, otherwise we would be finished.  Without loss of generality assume that agent $j$ is in the critical region of agent $i$.  First notice that:
  \begin{equation*}
    \frac{d}{dt}|{\bf x}_{i} - {\bf x}_{j}|^{2} = -2\langle(P_{\mathcal{C}_{i}}(\overline{{\bf x}}_{i} - {\bf x}_{i}), {\bf x}_{j} - {\bf x}_{i})\rangle - 2\langle(P_{\mathcal{C}_{j}}(\overline{{\bf x}}_{j} - {\bf x}_{j}), {\bf x}_{i} - {\bf x}_{j})\rangle.
  \end{equation*}
  % We can now proceed in two cases. First, suppose $\mathcal{C}_{i}=\{0\}$, i.e. no matter where agent $i$ will move, it will increase its distance with one of its 'critical neighbors'. Here $P_{\mathcal{C}_{i}}(\bar{\bf x}_{i}-{\bf x}_{i}) = 0$. Otherwise,
  Denote  ${\bf v}_i=P_{\mathcal{C}_{i}}(\bar{\bf x}_{i}-{\bf x}_{i})$. Since $j$ is in the critical region of agent $i$ by assumption, it must satisfy $\langle {\bf v}_i, {\bf x}_{j}-{\bf x}_{i} \rangle \geq 0$, therefore $\langle P_{\mathcal{C}_{i}}(\bar{\bf x}_{i}-{\bf x}_{i}), {\bf x}_{j}-{\bf x}_{i} \rangle \geq 0$. We can prove similarly that $\langle P_{\mathcal{C}_{j}}(\bar{\bf x}_{j}-{\bf x}_{j}), {\bf x}_{i}-{\bf x}_{j} \rangle \geq 0$ and therefore,
  \begin{displaymath}
    \frac{d}{dt} |{\bf x}_{i}-{\bf x}_{j}|^2 \;\leq\; 0.
  \end{displaymath}

  % % 1-D PROOF:
  % First notice that:
  % \begin{equation*}
  %   \frac{d}{dt} |x_{i}-x_{j}|^{2} = 2\mu_{i}(x_{i} - x_{j})(\bar{x}_{i}-x_{i}) + 2\mu_{j}(x_{j}-\bar{x}_{j})(x_{i}-x_{j})
  % \end{equation*}
  % We will show that $2\mu_{i}(x_{i} - x_{j})(\bar{x}_{i}-x_{i})\leq 0$ by proceeding in two cases.  The argument that $2\mu_{j}(x_{j}-\bar{x}_{j})(x_{i}-x_{j})\leq 0$ is analogous.
  %
  % \begin{Case}
  %   Assume that $ (x_{i} - x_{j})(\bar{x}_{i}-x_{i}) \leq 0$.  Then, since $\mu_{i} \geq 0$ we have that $2\mu_{i}(x_{i} - x_{j})(\bar{x}_{i}-x_{i})\leq 0$.
  % \end{Case}
  %
  % \begin{Case}
  %   Assume that $(x_{i} - x_{j})(\bar{x}_{i}-x_{i}) \geq 0$.  Then, since $(x_{j} - x_{i})(\bar{x}_{i}-x_{i}) \leq 0$, $x_{j}$ is in the critical region of $x_{i}$ and therefore $\mu_{i} = 0$ and so $2\mu_{i}(x_{i} - x_{j})(\bar{x}_{i}-x_{i})\leq 0$.
  % \end{Case}
  % \setcounter{Case}{0}

\end{proof}

The critical region acts as a "trap".  If the distance between two particles is less than or equal to $1$ but starts to increase they will eventually be within the critical distance of each other and their distance cannot increase any longer.

\begin{cor}\label{cor:preserve_connectivity}
  Suppose $|{\bf x}_i(t_0)-{\bf x}_j(t_0)|\leq1$. Then
  \begin{displaymath}
    |{\bf x}_i(t)-{\bf x}_j(t)|\leq1 \quad \text{for all} \quad t\geq t_0.
  \end{displaymath}
\end{cor}
\begin{proof}
  Suppose for the sake of contradiction that there exists $T>t_{0}$ such that $|{\bf x}_{i}(T) - {\bf x}_{j}(T)|>1$.  Since the dynamics are continuous there must exist an exit time $t_{1}$ satisfying $|{\bf x}_{i}(t_{1}) - {\bf x}_{j}(t_{1})| = 1$ and $|{\bf x}_{i}(t) - {\bf x}_{j}(t)| < 1$ for $t<t_{1}$.  Thus there must exist some $\delta > 0$ such that $1 - r \leq |{\bf x}_{i}(t) - {\bf x}_{j}(t)|\leq 1$ for $t_{1} - \delta < t < t_{1}$. Therefore, by Proposition \ref{prop:distance_decaying} $|{\bf x}_{i}(t) - {\bf x}_{j}(t)|$ is decaying on this time interval, a contradiction as $|{\bf x}_{i}(t_{1}) - {\bf x}_{j}(t_{1})| = 1$.
\end{proof}

So, adding the control indeed prevents situations like the ones presented in Figures \ref{fig:example_withWithout_ctrl} and \ref{fig:2D_timelapse} and preserves connectivity.

\begin{cor}\label{cor:preserve_connectivity_final}
  The dynamics \eqref{eq:no_one_left_behind} preserve connectivity, i.e. if the configuration $\{{\bf x}_i(t_0)\}_i$ is connected, then $\{{\bf x}_i(t)\}_i$ will be connected for any $t\geq t_0$.
\end{cor}

We will now examine how the preservation of connectivity enforced by \eqref{eq:no_one_left_behind} leads to a consensus.

\subsection{Emergence of a consensus}
We will now see that under the NOLB dynamics defined in \eqref{eq:no_one_left_behind}, connectedness of the initial condition is sufficient for convergence to a consensus.  In fact, this demonstrates that connectedness of the initial condition is equivalent to emergence of a consensus as we have previously noted that it is at least necessary.  We will examine the convergence in two cases.  In the general multidimensional case we are prevented from employing traditional ODE methods due to the discontinuous nature of the flow of the NOLB dynamics.  We find that the interplay between the contractive nature of the dynamics and the preservation of connectivity allows us to circumvent this difficulty and deduce that a connected initial configuration is sufficient for convergence to a consensus.  However we aren't able to say anything in general about the rate at which the dynamics converge to a consensus.  In the one dimensional case this information is available and we derive explicit rates of convergence.

\subsubsection{Multi Dimensional Case}
We first examine the general multi dimensional case.  Before proceeding to the main results we note that in the special case that $r_* = 1$ that connectivity of the initial condition is \textit{not} sufficient for convergence to a consensus; the large value of $r_*$ prevents some agents who are connected from exerting influence over each other.
\begin{example}\label{ex:counter_example}
  For simplicity we examine the one dimensional case.  Suppose that $r_* = 1$ and consider the initial configuration given by:
  \begin{equation*}
    {\bf x}_{1}(0) = 1, \; {\bf x}_{2}(0) = 2, \; {\bf x}_{3}(0) = 3, \; {\bf x}_{4}(0) = 4.
  \end{equation*}
  Here, ${\bf x}_{1}$ will move towards ${\bf x}_{2}$ and ${\bf x}_{4}$ will move towards ${\bf x}_{3}$. However, since $r_* = 1$, ${\bf x}_{2}$ and ${\bf x}_{3}$ may never move towards each other despite their connection as ${\bf x}_{1}$ and ${\bf x}_{4}$ will always be in their respective critical regions - see Figure \ref{fig:counter_example}.  Additionally we note that this example provides an illustration that the NOLB dynamics are \textit{not} continuous as the dynamics converge to a state that is not an equilibrium of the dynamics.
  % Indeed, if we rewrite the dynamics:
  % \begin{equation*}
  %   \vect{x}' =\vect{f}(\vect{x}) \quad \text{where}\quad \vect{f}(\vect{x}) = \begin{pmatrix} \overline{x}_{1} - x_{1} \\ \vdots \\ \overline{x}_{n} - x_{n} \\ \end{pmatrix}
  % \end{equation*}
  % Then, as the example above illustrates, for any $\delta\leq1$ we have that:
  % \begin{equation*}
  %   \vect{f}(\begin{pmatrix} 3+ \delta \\ 3 \\ 2 \\ 2-\delta \\ \end{pmatrix}) = \begin{pmatrix}  -\delta \\ 0 \\ 0 \\ \delta \\ \end{pmatrix} \rightarrow \vect{0}\;\text{as}\;\delta\rightarrow 0.
  % \end{equation*}
  % However we also have that:
  % \begin{equation*}
  %   \vect{f}(\begin{pmatrix} 3\\ 3 \\ 2 \\ 2\\ \end{pmatrix}) = \begin{pmatrix}  -2/3 \\ -2/3 \\ 2/3 \\ 2/3 \\ \end{pmatrix} \neq \vect{0},
  % \end{equation*}
  % so $\vect{f}$ is not continuous.

  \begin{figure}[ht]
    \centering
    \includegraphics{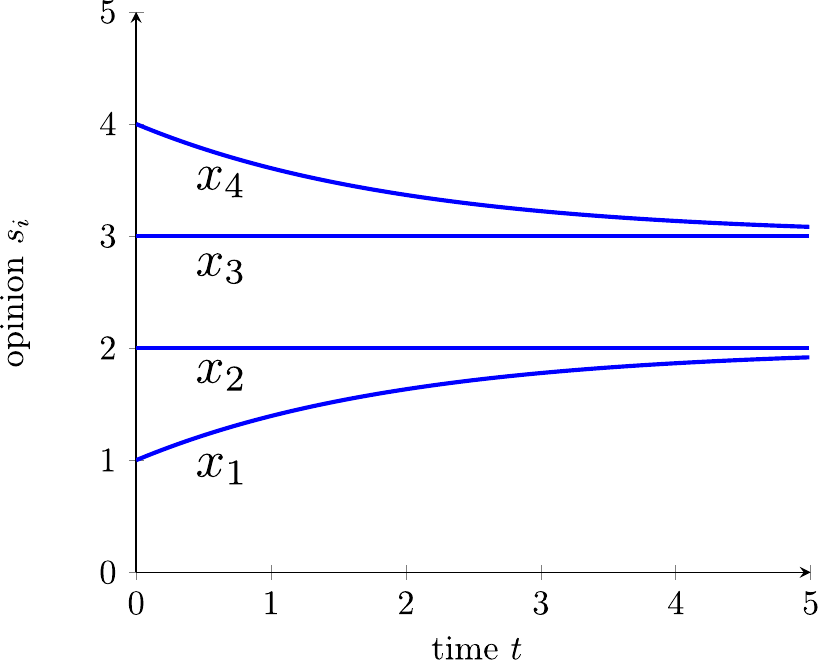}
    \caption{Preserving connectivity does not imply the convergence to a consensus. Here, when $r_*=1$, the extreme points $x_1$ and $x_4$ will converge towards $x_2$ and $x_3$ respectively. However, $x_2$ and $x_3$ cannot move since $x_1$ and $x_4$ are always in their respective critical regions. }
    \label{fig:counter_example}
  \end{figure}
\end{example}

We will now show that the situation described in Example \ref{ex:counter_example} is indeed a special case.  As long as $0 < r_* < 1$ we will see that the control is sufficient to guarantee consensus given that the initial configuration is connected.  Our main obstacle in this argument is discontinuities in the flow caused by the definition of the critical region  and the discontinuity in the interaction function, $\Phi$.  Since this rules out many of the standard tools in ODE theory, our argument will rely on the interplay between the contractive nature of the dynamics and the fact that they preserve connectivity of the configuration of agents.  Our strategy to show that the NOLB dynamics result in a consensus will be to show that the limiting configuration provided by Corollary \ref{cor:sequence}, $\vect{x}^{\infty}$, is a consensus.  In general this isn't sufficient to conclude that $\vect{x}(t)$ even converges, much less to a consensus.  However, the contractive nature of the dynamics allows us to say more.  % Notice that as a consequence of Proposition \ref{prop:contract} that the \textit{diameter}, $d(t) = \max_{ij} |x_{i}(t) - x_{j}(t)|$, of the configuration decays:
% \begin{equation*}
%   d(t_{2})\leq d(t_{1})\quad\text{for any $t_{2}\geq t_{1}$}.
% \end{equation*}
% Therefore if $\vect{x}^{\infty}$ is a consensus then we may conclude that $\vect{x}(t)$ converges to one as well as:
% \begin{align*}
    %     \lim_{t\rightarrow\infty}d(t) = d^{\infty} = 0.
    %   \end{align*}

We will proceed by contradiction and show that if $\textbf{x}^{\infty}$ is \textit{not} a consensus then it is possible to find a term in the sequence $\set{\vect{x}(t_{n})}$ provided by Corollary \ref{cor:sequence} that when taken as initial condition of the dynamics results in points of $\Omega^{\infty}$ being outside of the convex hull of the configuration at a finite time - a contradiction of Proposition \ref{prop:contract}.  We now prove the main result.

    %     \begin{lem}\label{lem:trap}
    %     Assume $\textup{Conv}(\set{x_{i}}_{i}))\subseteq\textup{Conv}(\set{y_{i}}_{i})$ and $\vect{y}$ evolves according to \eqref{eq:no_one_left_behind}.  Then, if $y_{i}$ is an extreme point of $\textup{Conv}(\set{y_{i}}_{i})$ we have that:
    %     \begin{align*}
    %     \frac{d}{dt}[\|x_{i} - y_{i}\|] \leq 0
    %   \end{align*}
    %     \end{lem}
    %     \begin{proof}
    %     See Appendix
    %     \end{proof}

\begin{theorem}\label{them:consensus}
  Assume $r_* < 1$ and that $\set{{\bf x}_{i}(0)}_{i}$ is connected. Then if $\set{{\bf x}_{i}(t)}_{i}$ evolves according to the NOLB dynamics \eqref{eq:no_one_left_behind} it will converge to a consensus. %the  we must have that the diameter $d(t)$ satisfies $\lim_{t\rightarrow\infty}d(t) = 0$.
\end{theorem}
\begin{proof}
  If $\set{{\bf x}_{i}}_{i}$ evolves according to the NOLB dynamics then by Corollary \ref{cor:sequence} (contractiveness) we know there exists a limiting configuration ${\bf X}^{\infty} = \set{{\bf x}_{i}^{\infty}}_{i}$ and a sequence of times $(t_{n})_{n}$ such that $t_{n}\rightarrow\infty$ and ${\bf x}_{i}(t_{n})\rightarrow {\bf x}_{i}^{\infty}$ as $n\rightarrow \infty$.  Note that since by assumption ${\bf X}(0)$ is connected we must have by Corollary \ref{cor:preserve_connectivity} (preservation of connectivity) that ${\bf X}^{\infty}$ is connected as well.  We denote by $p$ and $q$ the extreme points such that:
  \begin{align*}
    \|{\bf x}_{p}^{\infty} - {\bf x}_{q}^{\infty}\| = \max_{ij}\|{\bf x}_{i}^{\infty} - {\bf x}_{j}^{\infty}\|.
  \end{align*}
  Denote $\mathcal{N}_p^\infty$ as the set of neighbors of $p$. The main difficulty in the proof is to handle neighbors of ${\bf x}_{p}^{\infty}$ at a distance exactly $1$. We call them {\it extreme} neighbors and denote them by $\mathcal{E}_p^\infty$:
  \begin{eqnarray}
    \label{eq:neighbors}
    \mathcal{N}_p^\infty &=& \{j \;\;|\;\; \|{\bf x}_{j}^{\infty} - {\bf x}_{p}^{\infty}\| \leq 1\} \\
    \label{eq:extreme}
    \mathcal{E}_p^\infty &=& \{j \;\;|\;\; \|{\bf x}_{j}^{\infty} - {\bf x}_{p}^{\infty}\| = 1\}.
  \end{eqnarray}
  We are going to investigate 3 cases detailed in figure \ref{fig:proof_case_summary}.

  \begin{figure}[ht]
    \centering
    \includegraphics[width=.97\textwidth]{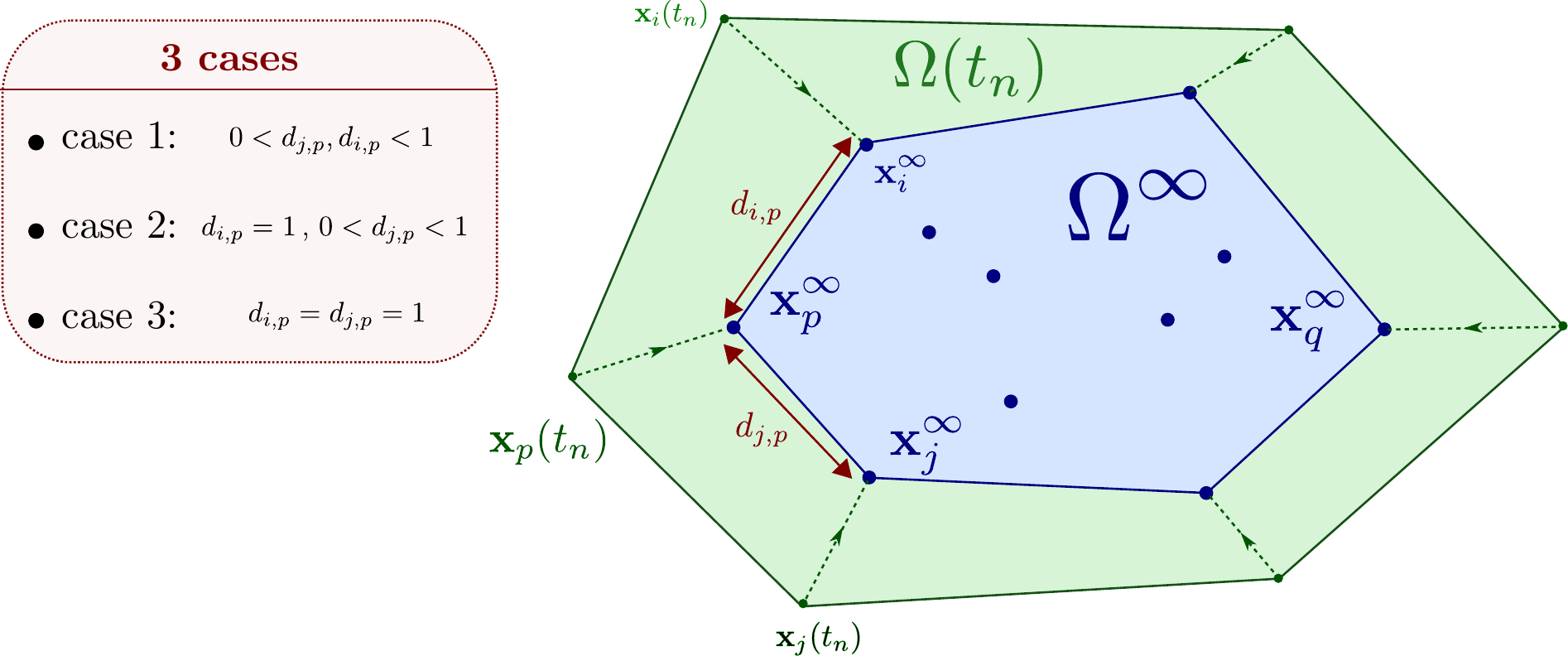}
    \caption{The convex hull $\Omega(t_n)$ has to converge to a limit configuration $\Omega^\infty$. The dynamics converge to a consensus if $\Omega^\infty$ is reduced to a single point which we prove by contradiction. We distinguish three cases of limit configuration $\Omega^\infty$  depending if the extreme point ${\bf x}_p^\infty$ has a so-called {\it extreme} neighbor $j$, i.e. $\|{\bf x}_p^\infty-{\bf x}_j^\infty\|=1$.}
    \label{fig:proof_case_summary}
  \end{figure}

  % We consider $p$ such that ${\bf x}_{p}(t_n) \stackrel{n \to +\infty}{\longrightarrow} {\bf x}_{p}^{\infty}$ and we denote its neighbors $\mathcal{N}_p$. The number of neighbors of $p$ evolves over time but due to preservation of connectivity (corollary \ref{cor:preserve_connectivity}), the set is actually strictly increasing. In all the following, we suppose that we look at time $t_n$ where the set $\mathcal{N}_p$ does not evolve anymore.

  \medskip

  $\bullet$ {\bf Case 1: no extreme neighbors $\mathcal{E}_p^\infty=\emptyset$.}

  \smallskip

  If all the neighbors of ${\bf x}_p^\infty$ are at distance $0$, then the dynamics converge to a consensus. Thus, we need to look at the case where there exists $j\in\mathcal{N}_p^\infty$ neighbor of $p$ at distance $0<\|{\bf x}_j^\infty-{\bf x}_p^\infty\|<1$ (the distance $1$ is excluded in Case 1).

  Denote $\Omega(t)$ and $\Omega^{\infty}$ the convex hull of $\set{{\bf x}_{i}(t)}_{i}$ and $\set{{\bf x}_{i}^{\infty}}$ respectively. Since the dynamics is contracting (Proposition \ref{prop:contract}), we must have $\Omega^{\infty} \subseteq \Omega(t)$ for any $t$. Take now a supporting hyperplane that is tangent to $\Omega^{\infty}$ at the extreme point ${\bf x}_{p}^{\infty}$ (see figure \ref{fig:proof_case1}). More specifically, we parametrize this supporting hyperplane using $\upvarphi_{p}:\Rn\rightarrow\R$ affine function
  \begin{align*}
    \upvarphi_{p}({\bf x}) = \inner{{\bf u}_{p}}{{\bf x}} + b
  \end{align*}
   where the vector ${\bf u}_p$ and constant $b$ are such that $\upvarphi_{p}({\bf x}^{\infty}_{p}) = 0$ and $\upvarphi_{p}({\bf x}) > 0$ for all ${\bf x}$ in $\Omega^{\infty}$ where ${\bf x}\neq {\bf x}^{\infty}_{p}$.  % In particular we have that:
  % \begin{align} \label{eq:special_normal}
      %       \inner{{\bf x}-{\bf x}^{\infty}_{p}}{{\bf u}_{p}} > 0
      %     \end{align}
      %       for all ${\bf x}$ in $\Omega^{\infty}$ where $x\neq x^{\infty}_{p}$.

  If there exists $j$ such that $0<\|{\bf x}_{p}^{\infty} - {\bf x}_{j}^{\infty}\|<1$, then the coefficient $a_{pj}^\infty$ satisfies:
  \begin{displaymath}
    a_{pj}^\infty=\frac{\Phi_{pj}^\infty}{\sum_{k=1}^N \Phi_{pk}^\infty} >0. %\;\;\;\text{ with }\; \Phi_{ij}^\infty=\Phi(\left | {\bf x}_{j}^\infty-{\bf x}_{i}^\infty \right |).
  \end{displaymath}
  Therefore, the local average $\overline{\bf x}_{p}^{\infty} = \sum_{k=1}^Na_{pk}^{\infty}{\bf x}_{j}^{\infty}$ is different from ${\bf x}_{p}^{\infty}$. Moreover since $\overline{\bf x}_{p}^{\infty} \in\Omega^\infty$, we deduce $\upvarphi_{p}(\overline{\bf x}_{p}^{\infty}) > 0$.

  \begin{figure}[ht]
    \centering
    \includegraphics[width=.5\textwidth]{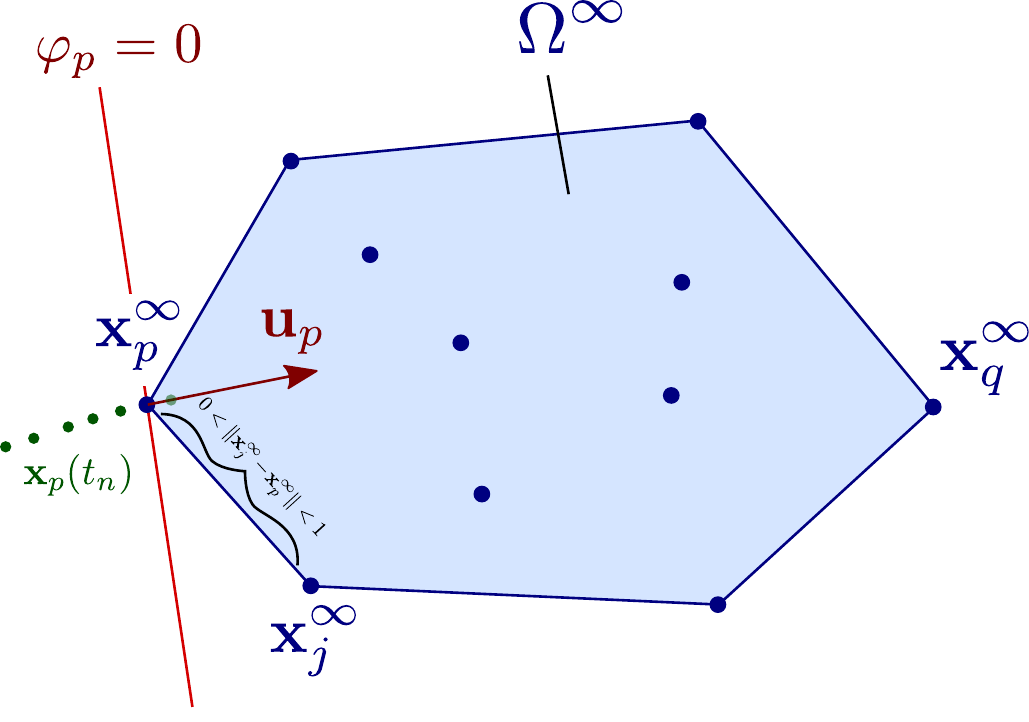}
    \caption{If the limit configuration $\{{\bf x}_k^\infty\}_k$ is not a consensus, the extreme point ${\bf x}_p(t_n)$ will eventually get inside the convex hull $\Omega^\infty$ which gives a contradiction.}
    \label{fig:proof_case1}
  \end{figure}

  We will get a contradiction if we can show that the sequence ${\bf x}_{p}(t_n)$ once closed to $\overline{\bf x}_{p}^{\infty}$ will {\it cross} the hyperplane $\{\varphi_p = 0\}$. Thus, we investigate the values of $\varphi_p({\bf x}_p(t_n))$. Its time evolution is given by:
  \begin{equation}
    \label{eq:start_computation}
    \frac{d}{dt}[\upvarphi_{p}({\bf x}_{p}(t_{n}))] = \inner{{\bf u}_{p}}{P_{\mathcal{C}_p(t_{n})}(\overline{\bf x}_{p}(t_{n}) - {\bf x}_{p}(t_{n}))}.
  \end{equation}
  To get rid of the projection operator, we notice that the projection is actually increasing the scalar product {\it if} ${\bf u}_p$ is in $\mathcal{C}_{p}(t_{n})$ thanks to the following lemma (the proof is in appendix).

  \begin{lem}\label{lem:project_help}
    Let $\set{{\bf x}_{i}}_{i}\subseteq\Rn$ and consider $C = \set{{\bf v} | \inner{{\bf v}}{{\bf x}_{i}}\geq 0\;\text{for all $i$}}$.  If ${\bf u}$ is in $C$ then:
    \begin{equation}
      \label{eq:increase_scalar_product}
      \inner{P_{C}({\bf x})}{{\bf u}}\geq\inner{{\bf x}}{{\bf u}} \quad \text{for all } {\bf x} \in \Rn.
    \end{equation}

  \end{lem}
It remains to show that ${\bf u}_{p}\in\mathcal{C}_p(t_{n})$. Notice that $\inner{{\bf x}-{\bf x}^{\infty}_{p}}{{\bf u}_{p}} \geq 0$ for all ${\bf x}$ in $\Omega^{\infty}$ and therefore ${\bf u}_{p} \in \mathcal{C}_{x_{p}^{\infty}}$.  Eventually this is true for the approximating sequence ${\bf x}_{p}(t_n)$ as well: ${\bf u}_{p} \in \mathcal{C}_{x_{p^{*}}(t_{n})}$ for $t_n$ large enough. Indeed, take any $k\neq p$. There are two cases. First case, ${\bf x}_k(t_n) \to {\bf x}_p^\infty$ and therefore $\|{\bf x}_k(t_n)-{\bf x}_p(t_n)\| \stackrel{n \to +\infty}{\longrightarrow} 0$, meaning that ${\bf x}_k(t_n)$ will not be in the critical region of agent $p$, i.e. $k\notin \mathcal{B}_{{\bf x}_p(t_n)}$ ($r_*<1$). In the second case, ${\bf x}_k(t_n) \to {\bf x}_k^\infty \neq {\bf x}_p^\infty$. Then
  \begin{align*}
    \inner{{\bf u}_{p}}{{\bf x}_{k}(t_{n}) - {\bf x}_{p}(t_{n})} &=\inner{{\bf u}_{p}}{{\bf x}_{k}(t_{n}) + {\bf x}^{\infty}_{k}-{\bf x}^{\infty}_{p}+{\bf x}^{\infty}_{p}-{\bf x}^{\infty}_{k}-{\bf x}_{p^{*}}(t_{n})} \\
                                                                 &= \inner{{\bf u}_{p}}{{\bf x}_{k}(t_{n})-{\bf x}^{\infty}_{k}} + \inner{{\bf u}_{p}}{{\bf x}^{\infty}_{p}-{\bf x}_{p^{*}}(t_{n})} + \inner{{\bf u}_{p}}{{\bf x}^{\infty}_{k} - {\bf x}_{p}^{\infty}} \\
                                                                 &\geq -2\norm{{\bf u}_{p}} \delta_n+ \inner{{\bf x}^{\infty}_{k}-{\bf x}^{\infty}_{p}}{{\bf u}_{p}}
  \end{align*}
  using Cauchy-Schwarz with $\delta_n = \max(\|{\bf x}_{k}(t_{n})-{\bf x}^{\infty}_{k}\|,\|{\bf x}_{p}(t_{n})-{\bf x}^{\infty}_{p}\|)$.  Since $\delta_n \stackrel{n \to +\infty}{\longrightarrow} 0$ and $\inner{{\bf x}^{\infty}_{k}-{\bf x}^{\infty}_{p}}{{\bf u}_{p}}>0$, we conclude that: $\inner{{\bf u}_{p}}{{\bf x}_{k}(t_{n}) - {\bf x}_{p}(t_{n})}>0$ for $t_n$ large enough, i.e. ${\bf u}_p\in \mathcal{C}_{p}(t_n)$.

  We can now continue our computation \eqref{eq:start_computation}:
  \begin{align}\label{eq:getinside_bound}
    \begin{split}
      \frac{d}{dt}[\upvarphi_{p}({\bf x}_{p}(t_{n}))] &= \inner{{\bf u}_{p}}{P_{\mathcal{C}_p(t_{n})}(\overline{\bf x}_{p}(t_{n}) - {\bf x}_{p}(t_{n}))} \\
      &\geq \inner{{\bf u}_{p}}{\overline{\bf x}_{p}(t_{n}) - {\bf x}_{p}(t_{n})} \\
      &= \varphi_p(\overline{\bf x}_{p}(t_{n}) - \varphi_p({\bf x}_{p}(t_{n})) \quad \rightarrow \quad \upvarphi_{p}(\overline{\bf x}_{p}^{\infty}) - \upvarphi_{p}({\bf x}_{p}^{\infty})>0
    \end{split}
  \end{align}
  by continuity of $\varphi_p$. Indeed, $\overline{\bf x}_{p}(t_{n}) \stackrel{n \to +\infty}{\longrightarrow} \to \overline{\bf x}_{p}^\infty$  since $\Phi$ has only a discontinuity at distance $1$ and the case 1 assumption avoids this possibility for any coefficients $a_{pk}$.

  We deduce a contradiction since we have two conflicting properties:
  \begin{eqnarray}
    &i)& \;\upvarphi_{p}({\bf x}_{p}(t_{n})) \to \upvarphi_{p}({\bf x}_{p}^\infty)=0 \quad \text{since }{\bf x}_{p}(t_{n}) \to {\bf x}_{p}^\infty, \\
    &ii)& \;\frac{d}{dt}[\upvarphi_{p}({\bf x}_{p}(t_{n}))]\geq \varphi_p(\overline{\bf x}_{p}^{\infty})>0 \text{ for large $t_n$}.
  \end{eqnarray}

  \medskip

  $\bullet$ {\bf Case 2: there exists an extreme neighbor (i.e. $\mathcal{E}_p^\infty \neq \emptyset$) AND there exists $j$ such that $0<\|{\bf x}_j^\infty-{\bf x}_p^\infty\|<1$.}

  \smallskip

  In this situation, agent $p$ might connects with a neighbor $k$ only at infinity and thus the local  average $\overline{\bf x}_p(t_n)$ might not converge to $\overline{\bf x}_p^\infty$ as $t_n\to+\infty$. But thanks to the non-extreme neighbor $j$, we are going to have a contradiction as in Case 1.

          %           There exists $j^{*}\in N_{p}^{\infty}$ such that $0<\|x_{p}^{\infty} - x_{j^{*}}^{\infty}\|<1$ - there is at least one nonextreme neighbor}

  The coefficient $a_{pj}$ is lower bounded at infinity since:
  \begin{align*}
    \liminf_{n\rightarrow \infty}a_{pj}(t_{n}) \geq \liminf_{n\rightarrow \infty} \frac{\Phi(\|{\bf x}_p(t_n)-{\bf x}_j(t_n)\|)}{N\cdot M} \geq \frac{m}{N\cdot M}>0
  \end{align*}
  where $m$ and $M$ are respectively the minimum and maximum of $\Phi$ on the interval $[0,1]$. This is enough to show that, as in case 1, the derivative $\frac{d}{dt}[\upvarphi_{p}({\bf x}_{p}(t_{n}))]$ is bounded below by a positive constant (for large $t_n$) leading to a contradiction. Indeed,
  \begin{eqnarray*}
    \liminf_{t_n\to+\infty} \inner{{\bf u}_p}{\overline{\bf x}_{p}(t_n)} &=& \liminf_{t_n\to+\infty}  \left(\sum_{k=1..N,\,k\neq j} a_{pk}\inner{{\bf u}_p}{{\bf x}_{k}(t_n)} \;\;+\; a_{pj} \inner{{\bf u}_p}{{\bf x}_{j}(t_n)} \right) \\
                                                                  & \geq & 0 + \frac{m}{N\cdot M} \inner{{\bf u}_p}{{\bf x}_j^\infty} >0.
  \end{eqnarray*}
  Thus, we conclude that:
  \begin{align*}
    \liminf_{n\rightarrow\infty} \big(\upvarphi_{p}(\overline{\bf x}_{p}(t_{n})) - \upvarphi_{p}({\bf x}_{p}(t_{n}))\big) \geq c > 0
  \end{align*}
  and we may apply the same argument as in Case 1.

  \medskip

  $\bullet$ {\bf Case 3: there exists an extreme neighbor (i.e. $\mathcal{E}_p^\infty \neq \emptyset$) AND for all $j\in \mathcal{N}_p^\infty$, $\|{\bf x}_{p}^{\infty} - {\bf x}_{j}^{\infty}\| = 0$ or $\|{\bf x}_{p}^{\infty} - {\bf x}_{j}^{\infty}\| = 1$.}

  \smallskip

   This is the most delicate case since the neighbors at distance {\it exactly} one might appear only asymptotically (i.e. at ``$t=infinity$''). However, the assumption is also helping: a neighbor $j$ of $p$ {\it must} converge to $p$ as we will see. More precisely, by the assumption of Case 3, any neighbor $j$ of ${\bf x}_p(t_n)$ must satisfy one of the two following scenario:
  \begin{enumerate}
  \item $\|{\bf x}_{j}(t_{n})  - {\bf x}_{p}(t_{n})\| \stackrel{n \to +\infty}{\longrightarrow}  0$
  \item $\|{\bf x}_{j}(t_{n})  - {\bf x}_{p}(t_{n})\| = 1$ for all $n$
  \end{enumerate}
  Indeed, if there exists a time $t_n$ such that $\|{\bf x}_{j}(t_{n})  - {\bf x}_{p}(t_{n})\|<1$, then it is impossible that $\|{\bf x}_{j}(t)  - {\bf x}_{p}(t)\|=1$ at a later time $t>t_n$ (Proposition \ref{prop:distance_decaying}). Since the limit $\|{\bf x}_{j}(t_{n})  - {\bf x}_{p}(t_{n})\|$ cannot be in $(0,1)$ due to the assumption of Case 3, it must converge to zero.

  To prove that a consensus emerges, we have to rule out scenario 2: neighbors of ${\bf x}_p(t_n)$ cannot stay at a distance exactly $1$. Notice that this is precisely what is happening in the counter-example of figure \ref{fig:counter_example}: $\|{\bf x}_3(t)-{\bf x}_2(t)\|=1$ for all $t$. But in this counter-example $r_*=1$ which is not the case in the present scenario.

  Let's proceed once again by contradiction and assume that there exists $j$ such that $\|{\bf x}_j(t_n)-{\bf x}_p(t_n)\|=1$ for all $t_n$. Denote $\mathcal{E}_p$ all the neighbors $j$ of $p$ satisfying this property. Under such assumption, we must have: $\frac{d}{dt} \|{\bf x}_j(t_n)-{\bf x}_p(t_n)\|^2 = 0$ and therefore:
  \begin{equation}
    \label{eq:proj_zero}
    \inner{P_{\mathcal{C}_p(t_n)}(\overline{\bf x}_{p}(t_{n}) - {\bf x}_{p}(t_{n}))}{{\bf x}_{j}(t_{n}) - {\bf x}_{p}(t_{n})} = 0.
  \end{equation}
  A key observation is to notice that the desired velocity $\overline{\bf x}_{p}(t_{n})-{\bf x}_{p}(t_{n})$ converges to the average over the neighbors in $\mathcal{E}_p$. Indeed, we rewrite:
  \begin{displaymath}
    \overline{\bf x}_{p}(t_{n})-{\bf x}_{p}(t_{n}) = \sum_{j=1}^N a_{pj}(t_n)\big({\bf x}_{j}(t_{n})-{\bf x}_{p}(t_{n})\big).
  \end{displaymath}
  If $j \notin \mathcal{E}_p$, either $a_{pj}(t_n)=0$ ($j$ not a neighbor of $p$) or ${\bf x}_{j}(t_{n})-{\bf x}_{p}(t_{n}) \stackrel{n \to +\infty}{\longrightarrow} 0$. Thus,
  \begin{equation}
    \label{eq:tp}
    \overline{\bf x}_{p}(t_{n})-{\bf x}_{p}(t_{n}) \stackrel{n \to +\infty}{\longrightarrow} \sum_{j \in \mathcal{E}_p} a_{pj}^\infty\big({\bf x}_{j}^\infty-{\bf x}_{p}^\infty\big).
  \end{equation}
  Moreover, $a_{pj}^\infty=c>0$ for all $j\in\mathcal{E}_p$ since  $\Phi(\|{\bf x}_{j}(t_n)-{\bf x}_{p}(t_n)\|) = \Phi(1)>0$.

  We can now pass to the limit in  \eqref{eq:proj_zero}:
  \begin{equation}
    \label{eq:tp2}
    \inner{P_{C_p^\infty}\big(\sum_{k \in \mathcal{E}_p} c({\bf x}_{k}^\infty-{\bf x}_{p}^\infty)\big)}{{\bf x}_{j}^\infty - {\bf x}_{p}^\infty} = 0,
  \end{equation}
  where $C_p^\infty$ is the critical region defined by:
  \begin{equation}
    \label{eq:C_p_inf}
    \mathcal{C}_{p}^\infty = \{{\bf v}\in\mathbb{R}^d\,|\;(  {\bf v}, {\bf x}_{j}^\infty-{\bf x}_{p}^\infty )\geq 0\quad \forall  j\in \mathcal{E}_p \}.
  \end{equation}
  Notice that the definition of $\mathcal{C}_{p}^\infty$ only includes extreme neighbors (i.e. $j\in \mathcal{E}_p$). Indeed, the other neighbor ${\bf x}_k$ converges to ${\bf x}_p$ as $t_n\to+\infty$. Since $r_*<1$, ${\bf x}_k(t_n)$ is no longer in the critical region $\mathcal{B}_{{\bf x}_p(t_n)}$ for $t_n$ large enough.

  Summing the previous expression over the neighbors $j$ in $\mathcal{E}_p$ gives:
  \begin{equation}
    \label{eq:tp3}
    \inner{P_{C_p^\infty}({\bf v})}{{\bf v}} = 0 \qquad \text{with} \quad {\bf v} = \sum_{j \in \mathcal{E}_p} ({\bf x}_{j}^\infty-{\bf x}_{p}^\infty).
  \end{equation}
  Since $C_{C_p^\infty}$ is a convex cone, we deduce that:
  \begin{equation}
    \inner{P_{C_p^\infty}({\bf v})}{P_{C_p^\infty}({\bf v})} = \inner{P_{C_p^\infty}({\bf v})}{{\bf v}},
  \end{equation}
  and therefore $P_{C_p^\infty}({\bf v})=0$.

  To get a contradiction, we use once again ${\bf u}_p$ define from the supporting hyperplane at ${\bf x}_p^\infty$ (see figure \ref{fig:proof_case1}), we have:
  \begin{equation}
    \label{eq:tp4}
    \inner{P_{C_p^\infty}({\bf v})}{{\bf u}} \geq \inner{{\bf v}}{{\bf u}}
  \end{equation}
  since ${\bf u}$ is in the cone $C_p^\infty$. We deduce:
  \begin{equation}
    \inner{P_{C_p^\infty}({\bf v})}{{\bf u}} \geq  \sum_{j\in \mathcal{E}_p} \inner{{\bf x}_j^\infty}{{\bf u}} >0,
  \end{equation}
  since ${\bf x}_j^\infty$ is strictly inside the convex hull $\Omega^\infty$ ($p$ being an extreme point). This proves that $P_{C_{{\bf x}_p^\infty}}({\bf v}) \neq 0$ and conclude the proof.
\end{proof}

\subsubsection{1-D Case}
We now investigate consensus in the special case that the dynamics are one dimensional.  We know from the previous section that consensus occurs however in this case we will also be able to quantify the rate at which this consensus emerges.  Our main tool will be estimates on the diameter $d(t)$.  Recall from Remark \ref{rem:1D_special_case} that in one dimension Model \eqref{model:no_one_left_behind} reduces to Model \eqref{model:1D_NOLB} - we will use the notation in the latter in the following. We will see that the control $\mu_{i}$ causes the convergence to occur in two stages.  Clearly if a consensus emerges there must exist a time $\tau$ after which all agents are directly interacting with each other, that is $|x_{i}(t) - x_{j}(t)| \leq 1$ for any $i$ and $j$ and $t>\tau$; in this case the network on which agents interact is \textit{fully connected}.  Before this time there necessarily exist pairs of agents who do not interact but are merely connected by a path.  We will first examine the case where all agents are interacting - in this case the dynamics converge towards a consensus at an exponential rate that depends on the extreme values of the interaction function.
\begin{prop}\label{ppo:exp_decay}
  Suppose $d(0)\leq 1$.  Then:
  \begin{equation}
    \label{eq:exponential_rate}
    d(t) \leq d(0) e^{-\frac{m}{M}\cdot t}
  \end{equation}
  where $m = \min_{x\in[0,1]} \Phi(x)$ and $M = \max_{x\in [0,1]} \Phi(x)$.
\end{prop}
\begin{proof}
  Fix $t$ and denote $p$ and $q$ such that $d(t) = |x_{p} - x_{q}|$.  Notice that since $p$ and $q$ are the two agents with \textit{extreme} opinions we must have that $\mu_{p} = \mu_{q} = 1$ as they cannot have any agents in their critical regions.  We aim to get a bound on $(d^{2})'$ in terms of $d^{2}$ in order to apply Gronwall's lemma.  By the Cauchy-Schwarz inequality we have that:
  \begin{equation}\label{eq:bound1_prop2}
    \begin{split}
      \frac{d}{dt}[d^{2}(t)] = 2(\dot{x}_{p} - \dot{x}_{q}, x_{p}-x_{q}) &= 2\Big((\overline{x}_{p} - \overline{x}_{q})(x_{p}-x_{q})- (x_{p}-x_{q})^{2}\Big)\\
      &\leq 2\Big(|\overline{x}_{p} - \overline{x}_{q}||x_{p}-x_{q}| - |x_{p}-x_{q}|^{2}\Big).\\
    \end{split}
  \end{equation}

  To obtain the bound we desire all that remains is to bound $|\overline{x}_{p} - \overline{x}_{q}|$ above by a constant multiple of $|x_{p}-x_{q}|$. We will exploit the fact that the local averages $\overline{x}_{p}$ and $\overline{x}_{q}$ must be inside the convex hull of opinions and therefore since agents $p$ and $q$ are the agents with the most extreme opinions the difference between their local averages must be smaller than the difference between their opinions.  Denote $\eta_{i} = \text{min}(a_{pi}, a_{qi})$ and notice that since $\Phi$ is bounded and $d(0)\leq 1$ we must have that
  $\eta_{i} \geq \frac{m}{MN}$ where $m$ and $M$ are given by \eqref{eq:m_and_M}.  Notice that:
  \begin{equation}\label{eq:squiggle1_prop2}
    \begin{split}
      \left|\overline{x}_{p} - \overline{x}_{q}\right| &= \left|\sum_{i = 1}^{N}a_{pi}x_{i} - \sum_{i = 1}^{N}a_{qi}x_{i}\right|  = \left|\sum_{i = 1}^{N}(a_{pi} - \eta_{i})x_{i} - \sum_{i =1}^{N}(a_{qi}-\eta_{i})x_{i}\right|\\
      & =:\left|\sum_{i = 1}^{N}\widetilde{a}_{pi}x_{i} - \sum_{i = 1}^{N}\widetilde{a}_{qi}x_{i}\right|.
    \end{split}
  \end{equation}
  Let $\boldsymbol{\eta} = \sum_{i=1}^{N}\eta_{i}$. We have by \eqref{eq:squiggle1_prop2} that:
  \begin{equation}\label{eq:squiggle2_prop2}
    \left|\overline{x}_{p} - \overline{x}_{q}\right|  = (1-\boldsymbol{\eta})\left|\sum_{i = 1}^{N}\frac{\widetilde{a}_{pi}x_{i}}{1-\boldsymbol{\eta}} - \sum_{i = 1}^{N}\frac{\widetilde{a}_{qi}x_{i}}{1-\boldsymbol{\eta}}\right|
    := (1-\boldsymbol{\eta})\left|\widetilde{x}_{p} - \widetilde{x}_{q}|\right.
  \end{equation}
  Notice that
  \begin{equation}\label{eq:convex_prop2}
    \sum_{i = 1}^{N}\frac{\widetilde{a}_{qi}}{1-\boldsymbol{\eta}} = \sum_{i = 1}^{N}\frac{\widetilde{a}_{pi}}{1-\boldsymbol{\eta}} = 1,
  \end{equation}
  and therefore we must have that $\widetilde{x}_{p}$ and $\widetilde{x}_{q}$ are in the convex hull of $\set{x_{i}}_{i=1}^{N}$ so necessarily we have that $|\widetilde{x}_{p} - \widetilde{x}_{q}|\leq |\overline{x}_{p} - \overline{x}_{q}|$.
  Therefore by \eqref{eq:squiggle2_prop2} we have that $|\overline{x}_{p} - \overline{x}_{q}| \leq (1- \boldsymbol{\eta})|x_{p} - x_{q}|$ which by \eqref{eq:bound1_prop2} implies:
  \begin{equation}\label{eq:prop2_finalbound}
    \begin{split}
      \frac{d}{dt}[d^{2}(t)] &\leq 2\Big((1- \boldsymbol{\eta})|x_{p} - x_{q}||x_{p}-x_{q}| - |x_{p}-x_{q}|^{2}\Big)\\
      & = -2\boldsymbol{\eta}|x_{p} - x_{q}|^{2}.
    \end{split}
  \end{equation}
  Finally, since by definition of the weights $a_{iq}, a_{ip}$ and $\boldsymbol{\eta}$ we have that $\boldsymbol{\eta}\geq \frac{m}{M}$ we can conclude using \eqref{eq:prop2_finalbound} that:
  \begin{equation}\label{eq:prop2_conclude}
    \frac{d}{dt}[d^{2}(t)] = 2d(t)d'(t) \leq -\frac{m}{M}d(t).
  \end{equation}
  An application of Gronwall's lemma provides the final result.
\end{proof}

So, in the case that all agents are interacting, i.e. that $|x_{i}(t) - x_{j}(t)| \leq 1$ for any $i$ and $j$, a consensus is reached exponentially fast at a rate that depends on the maximum and minimum values of the interaction function.  However, starting from an initial condition that is connected does not mean that all agents are directly interacting; any two agents are merely connected by a path.  We now examine the rate of convergence for $t<\tau$, i.e. \textit{before} all agents are interacting.

\begin{theorem}
  \label{thm:cv_1d}
  Suppose the initial condition $\{x_i(0)\}_i$ is connected and let $\eta=\frac{m}{M\cdot N}$. There exist $\delta>0$ and $T>0$ such that while $d(t)\geq1$, we have:
  \begin{equation}
    \label{eq:linear_decay}
    d(t)\leq d(0) + \eta\delta(N\cdot T-t).
  \end{equation}
        %         1=d_0+c_1-c_2 T  (d_0+c_1-1)/c_2
  Thus, after $t\geq N\cdot T+\frac{d_0-1}{\eta\delta}$, the diameter is converging exponentially fast toward zero.
\end{theorem}
\begin{proof}
  Denote $p$ and $q$ such that $d(t)=x_q(t)-x_p(t)$. Suppose $d(t)>1$, the case $d(t)\leq1$ has been treated in proposition \ref{ppo:exp_decay}. We analyze the behavior of $p$ and $q$ separately as they do not affect each other. We fix the two constant $\delta>0$ and $T>0$ satisfying the two technical conditions:
  \begin{eqnarray}
    \label{eq:delta_T_eq1}
    \delta+2T &\leq& \min(r_*,1-r_*) \\               % |x_p2-x_p1|\geq1-r_*, |x_p-x_p1|\leq1-r_*
    \label{eq:delta_T_eq2}
    \frac{T^2}{2N^2} \eta\big(\eta (1-\delta-2T)-2N(\delta+2T)\big) &\geq& \delta. % x_p connects with x_p2
  \end{eqnarray}
  It is always possible to find such constant (see lemma \ref{lem:delta_T}).

  We split the study of $\dot{x}_p$ at a given time $t_*$ in two cases.

      %       In the first case, we can find  find a lower bound for the velocity of $p$ and thus deduce a linear decay rate for $d(t)$.\\
      %       However, in the second case, we are not able to obtain a lower bound for $p$'s velocity. Instead, we prove that the agent $p$ will join its {\it nearest unconnected agent} (denoted $p_{2}$) after a time $T$ (upper bound). By induction, agent $p$ will repeat this procedure and finally achieve consensus.

  $\bullet$ {\bf Case 1}: suppose there exists $p_1$ such that $\delta\leq|x_p(t_*)-x_{p_1}(t_*)|\leq1$. \\
In this case, we can easily deduce a lower bound for the speed of $p$:
  \begin{displaymath}
    \dot{x}_{p} = \sum_{j=1}^{N}(x_{j}-x_{p})a_{jp} \geq\ (x_{p_{1}}-x_{p})a_{p,p_{1}}\geq \delta  \eta
  \end{displaymath}
  with $\eta = \frac{m}{M\cdot N}$. Thus, since we have on the other end $\dot{x}_q\leq0$, we deduce that the diameter is decaying at a minimum speed $\delta\eta$.

  $\bullet$ {\bf Case 2}:  $|x_p(t_*)-x_i(t_*)|<\delta$ or $|x_p(t_*)-x_i(t_*)|>1$ for any $i\neq p$. \\
  In other words, all the neighbors of $p$ are at a distance less than $\delta$. We cannot find a lower bound for $\dot{x}_p$ anymore. Instead, we show that a {\it neighbor of a neighbor} (denoted $p_2$) will become a new neighbor of $p$ in a finite time less than $T$.\\
  Since connectivity is preserved, there exists $p_1$ and $p_2$ such that: $p\sim p_1$, $p_1\sim p_2$ and $p \not \sim p_2$. Therefore, we deduce (see figure )
  \begin{displaymath}
    |x_p(t_*)-x_{p_1}(t_*)|\leq\delta , \quad |x_{p_1}(t_*)-x_{p_2}(t_*)|\leq1, \quad |x_{p}(t_*)-x_{p_2}(t_*)|>1.
  \end{displaymath}
  By triangular inequality, we deduce that $|x_{p_1}(t_*)-x_{p_2}(t_*)|\geq1-\delta$. Let us show that after time $T$, we have $|x_{p}(t_*+T)-x_{p_2}(t_*+T)|\leq1$.

  \begin{figure}[ht]
    \centering
    \includegraphics[width=.6\textwidth]{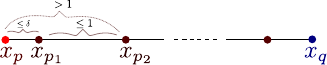}
    \caption{Situation in the case 2. The extreme point $x_p$ needs $x_{p_2}$ the neighbor of its neighbor $x_{p_1}$ to be pushed further to the right.}
    \label{fig:proof_1d}
  \end{figure}

  First, we show that $\dot{x}_{p_2}\leq0$ thanks to $x_{p_1}$. During the time interval $t\leq T$, as the velocity $|\dot{x}_i|\leq 1$ for any $i$, we have:
  \begin{displaymath}
    |x_{p_1}(t_*+t)-x_{p_2}(t_*+t)| \geq |x_{p_1}(t_*)-x_{p_2}(t_*)| - 2\cdot t \geq 1-\delta -2T \geq1-r_*,
  \end{displaymath}
  by the assumption \eqref{eq:delta_T_eq1}. As the consequence, $x_{p_1}$ is always in the critical region of the agent $p_2$. Thus, $x_{p_2}$ can only move left, which implies $\dot{x}_{p_2}\leq0$.

  Second, we show that $x_p$ increases by at least $\delta$ during the period $T$ which will imply that $p\sim p_2$. The main idea is to show that $x_{p_1}$ is going to {\it pull out} $x_p$. To prove it, we compute
  \begin{eqnarray*}
    \dot{x}_{p} &=& \mu_{p}(\bar{x}_{p}-x_{p}) = 1\cdot\sum_k a_{p,k}(x_k-x_p) \\
                &=& a_{p,p_1}(x_{p_1}-x_p) + \sum_{k\neq p_1} a_{p,k}(x_k-x_p) \;\;\geq\; \frac{\eta}{N} (x_{p_1}-x_p) + 0,
  \end{eqnarray*}
  since $x_k-x_p\geq0$. We now need to find a lower bound for $x_{p_1}-x_p$. With this aim, we compute the time derivative:
  \begin{eqnarray*}
    \dot{x}_{p_1}-\dot{x}_p &=& \mu_{p_1}(\bar{x}_{p_1}-x_{p_1}) - \mu_{p}(\bar{x}_{p}-x_{p}) \\
                            &\geq& (\bar{x}_{p_1}-x_{p_1}) - (\delta+T).
  \end{eqnarray*}
                                %                                 \bar{x}_{p}-x_{p} \leq x_{p_1}-x_p \leq  \delta+T
  Here, we use that $\dot{x}_{p_1}$ is necessarily positive as $p_2$ is in its critical region, thus we have a lower bound by replacing $\mu_{p_1}$ by $1$. We now also suppose that $x_p$ has not connected with $p_2$ (otherwise there is no need to go further) and therefore its neighbors are at a distance bounded by $\delta+T$ on the time interval $[t_*,t_*+T]$.

  Following the same inequality as for $\dot{x}_{p}$, we deduce that:
  \begin{eqnarray*}
    \bar{x}_{p_1}-x_{p_1} &=& a_{p_1,p_2}(x_{p_2}-x_{p_1}) + \sum_{k\neq p_2} a_{p_1,k}(x_k-x_{p_1}) \\
                          &\geq& \frac{\eta}{N} (1-\delta-2T) + (x_p-x_{p_1}) \;\;\geq\; \frac{\eta}{N} (1-\delta-2T) - (\delta+T).
  \end{eqnarray*}
  Therefore,
  \begin{displaymath}
    \dot{x}_{p_1}-\dot{x}_p \geq \frac{\eta}{N} (1-\delta-2T) - 2(\delta+T).
  \end{displaymath}
  We conclude that:
  \begin{eqnarray*}
    x_{p_1}(t_*+t)-x_p(t_*+t) &\geq& x_{p_1}(t_*)-x_p(t_*) + \int_0^t \frac{\eta}{N} (1-\delta-2T) - 2(\delta+T) \,\mathrm{d} t \\
                      &\geq& 0 + \left(\frac{\eta}{N} (1-\delta-2T) - 2(\delta+T)\right)t.
  \end{eqnarray*}
  Coming back to $x_p$, we obtain:
  \begin{eqnarray*}
    x_p(t_*+T)-x_p(t_*) &\geq& \int_0^T \frac{\eta}{N} (x_{p_1}(t_*+t)-x_p(t_*+t))\,\mathrm{d}t \\
                  &\geq& \frac{\eta}{N}\left(\frac{\eta}{N} (1-\delta-2T) - 2(\delta+T)\right) \frac{T^2}{2}\;\;\geq\;\delta,
  \end{eqnarray*}
  using \eqref{eq:delta_T_eq2}. Therefore, at time $t_*+T$, we have $x_{p_2}(t_*+T)-x_p(t_*+T)\leq1$, and thus $p\sim p_2$.

  \bigskip

  To conclude, since there is only a finite number of particles $N$, situations as in {\it case 2} can only appear a finite number of time (less than $N$ times) and thus spread over a period less than $N\cdot T$. Thus, outside these periods, the decay of $d(t)$ satisfies $\dot{d}\leq-\eta\delta$ leading to the upper-bound \eqref{eq:linear_decay} which concludes the proof.
\end{proof}

\begin{figure}[ht]
  \centering
  \includegraphics[scale=.68]{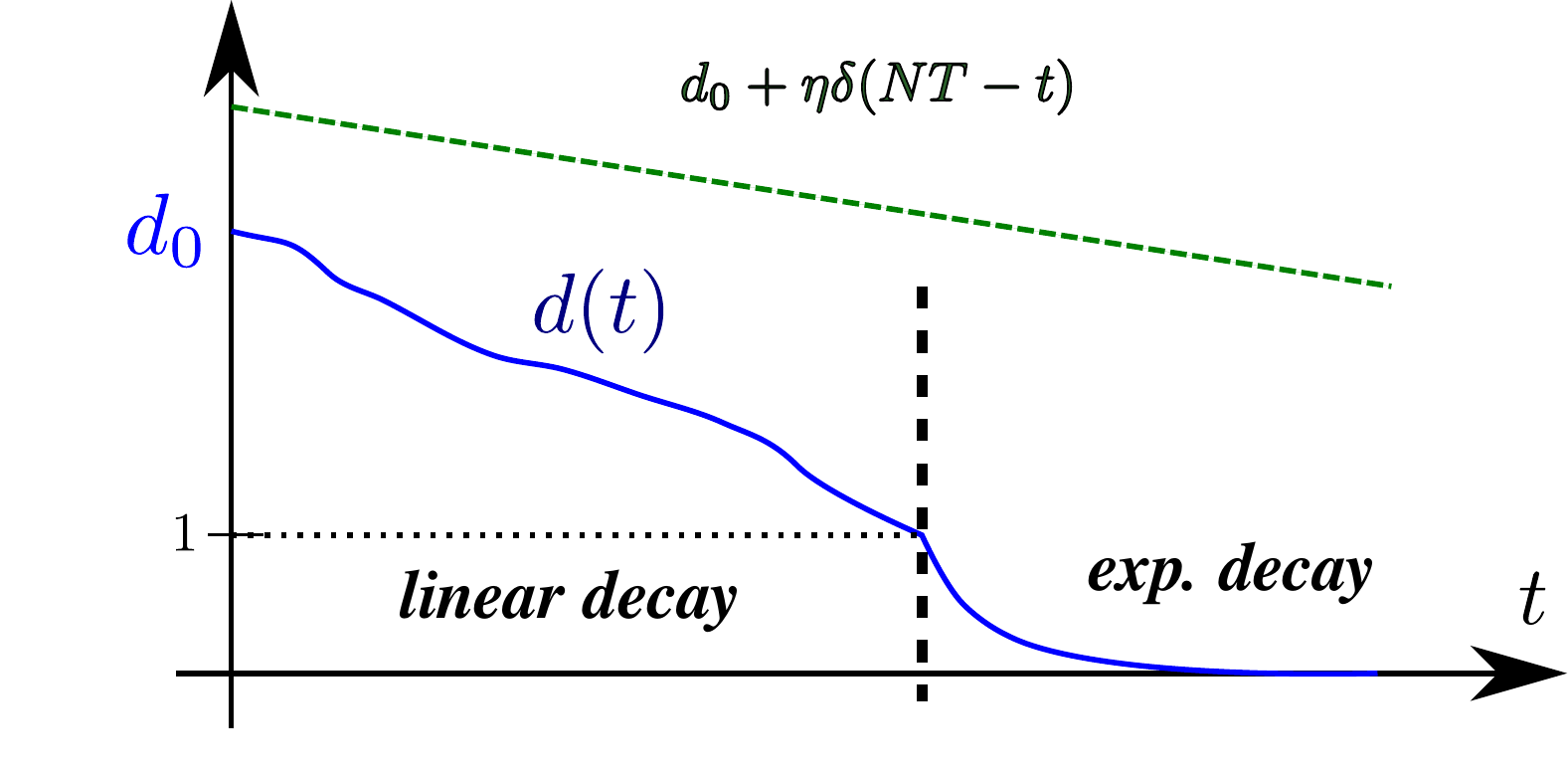}
  \caption{The decay of the diameter $d(t)$ is first linear and then exponential after the diameter $d(t)$ becomes less than $1$.}
  \label{fig:decay_d}
\end{figure}

\begin{remark}
   Notice that we could not derive an explicit decay rate in the multi-dimension case. Indeed, in one dimension, we exploited the property that the  left behind dynamics preserve the ordering of the agents and in particular that the diameter forming agents are the same agents for all time (i.e. $p$ and $q$ are time independent). In multiple dimensions it is possible for the diameter forming agents to change - this prevents us from applying the techniques used in one dimension to multiple dimensions.
\end{remark}

\subsection{Diameter decay: numerical experiments}

From the previous results, we can see that convergence to a consensus occurs in two stages. Before the diameter $d$ becomes less than $1$ the convergence is at least linear, after $d$ is less than $1$ the convergence becomes exponential.  The estimation of the convergence rates provided by Theorem \ref{thm:cv_1d} is fairly rough. We would like to explore numerically if in practice the decay of the diameter is faster.

We perform $100$ realizations of the $NOLB$ dynamics with initial conditions for the configuration ${\bf x}_i$ taken from a uniform distribution on the interval $[0,10]$ (1D simulation). We then compute the evolution of the diameter $d(t)=\max_{i,j}|x_i(t)-x_j(t)|$ for all the realizations. In Figure \ref{fig:1D_decay_d}-left, we plot the decay of the 'median' of the diameter (red) along with the slowest and fastest decay (dashed blue). To measure the disparity of $d(t)$, we also plot the $5\%$ and $95\%$ quantile. We observe two phases in the decay of the diameter: initially $d(t)$ decays quickly and then starts to slow down until it reaches the distance $1$ when it decays exponentially fast. We notice that there are a large variation between the different realizations. Indeed, if we denote $\tau$ the stopping time at which $d$ reaches $1$:
\begin{equation}
  \label{eq:tau}
  \tau=\min_{t\geq 0}\{d(t)\leq 1\},
\end{equation}
then we observe that $\tau$ varies between $22$ time units (fastest realization) and $50$ time units(slowest realization). Thus, finding a sharp decay rate for the NOLB dynamics seems challenging.

 Additionally, we would like to explore how the the radius, $r_{*}$, of the critical region affects the convergence.  Naively one might expect that lower values of $r_{*}$ result in faster convergence as agents are more free to move.  However, we find that this is not the case; intuitively since lower values of $r^{*}$ allow for more freedom of movement agents can have a higher degree of clustering which hinders the emergence of a consensus.  For each value of $r_{*}$ ranging from $0$ to $1$, we run $100$ simulations of the NOLB dynamics and estimate the stopping time $\tau$ \eqref{eq:tau}.  The results are plotted in Figure \ref{fig:1D_decay_d}-right.  We find that indeed $\tau$ is lower for higher values of $r_{*}$.  Interestingly, this effect seems to become less prominent for $r_{*}\geq 0.2$.

  \begin{figure}[ht]
    \centering
    \includegraphics[width=.47\textwidth]{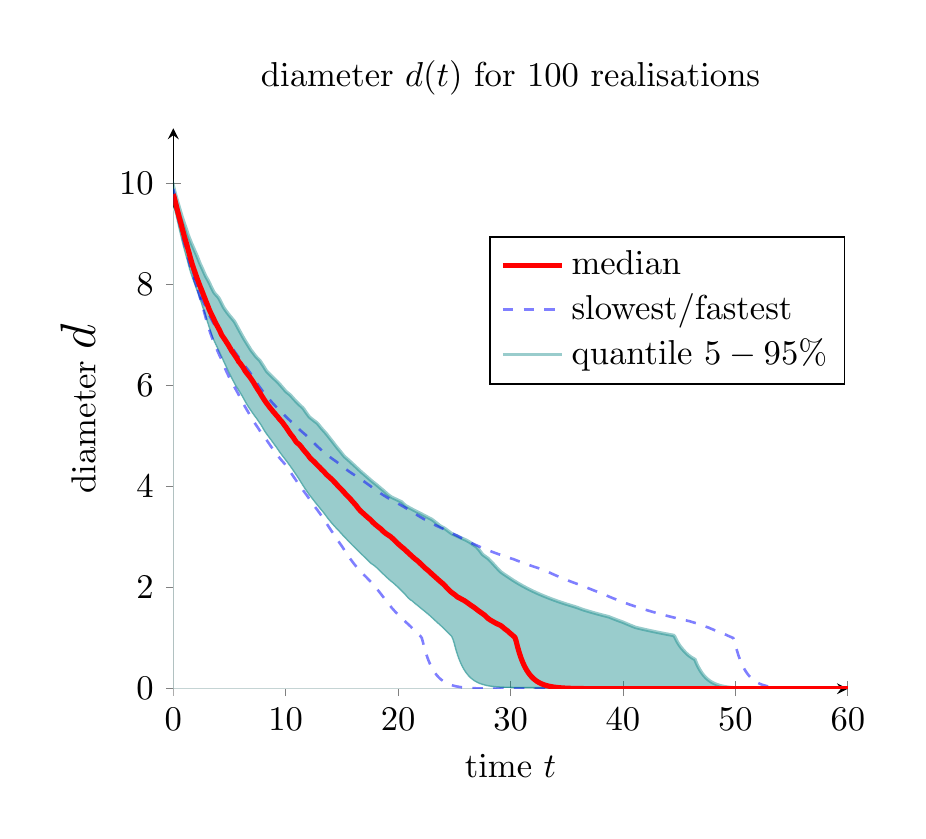}
    \includegraphics[width=.47\textwidth]{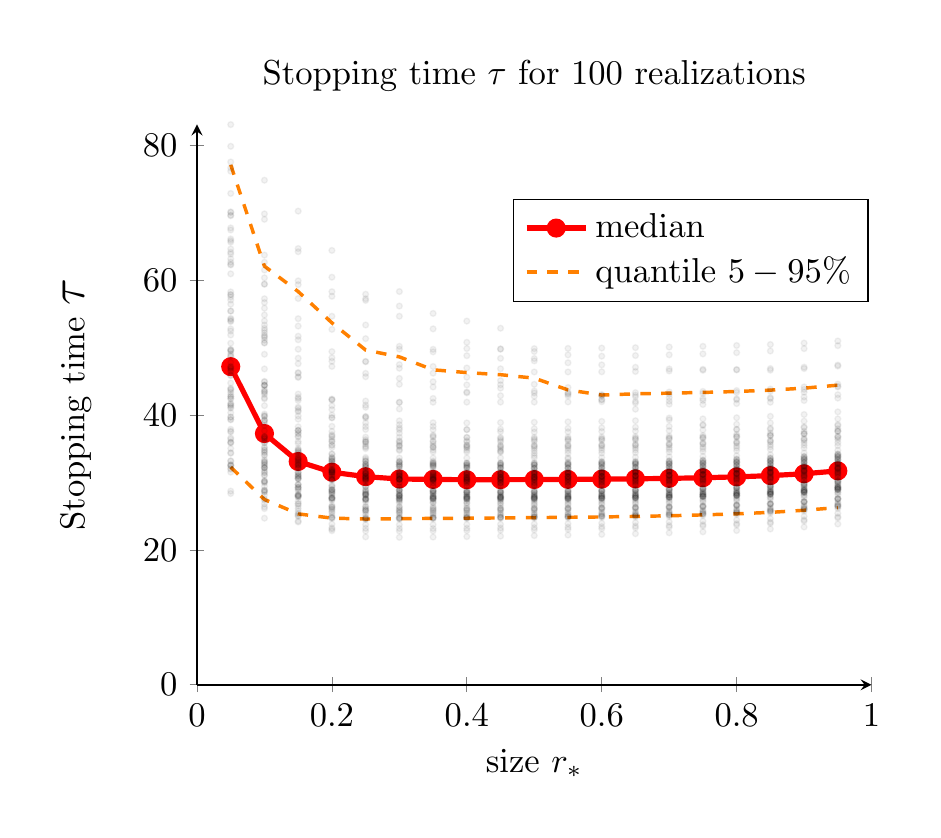}
    \caption{{\bf Left:} diameter $d(t)$ over time for $100$ realizations (quantile representation).  {\bf Right:} stopping time $\tau$ \eqref{eq:tau} depending on the size of the critical region $r_*$.}
    \label{fig:1D_decay_d}
  \end{figure}

\section{Relaxed no one left behind}

In this section we will investigate whether it is possible to weaken the constraints imposed by the NOLB dynamics and still maintain convergence to a consensus.  The critical ingredient in the argument used to show the convergence to consensus of the NOLB dynamics was the preservation of connectivity of the entire configuration of agents.  However, the dynamics introduced in \eqref{eq:no_one_left_behind} preserve connectivity between individual \textit{agents} by Proposition \ref{prop:distance_decaying} - once two agents begin interacting they continue to do so throughout the evolution of the dynamics as each agent "takes care of" every agent in its critical region.  However while this is clearly sufficient to preserve connectivity of the whole configuration, it isn't necessary.  Instead of maintaining direct connectivity between agents we need only to maintain a \textit{path} between them.

Instead of preventing all disconnections as in the NOLB dynamics, we could allow individual agents to disconnect as long as they remain connected to a mutual neighbor  - an agent does not have to "take care" of an agent in its behind region if one of its neighbors is already doing so. This intuition can be made rigorous via a description using the behind graph.  We name the resulting dynamics \textit{relaxed} no one left behind as we remove constraints while maintaining (global) connectivity. It is however unclear whether the new dynamics will lead to a faster convergence to consensus.

\subsection{Model introduction}

Before introducing the model, we formally define what we mean by a relaxed behind graph. We recall that we denote by $G=(V,E)$ and $G^{\mathcal{B}}=(V,E^{\mathcal{B}})$ respectively the interacting graph and the behind graph. $V=\{1\dots N\}$ is the set of all $N$ opinions while $E$ and $E^{\mathcal{B}}$ are the edge sets defined as:
\begin{itemize}
\item $(i,j)\in E$ if $\|{\bf x}_i-{\bf x}_j\|\leq 1$,
\item $(i,j)\in E^{\mathcal{B}}$ if $j\in \mathcal{B}_i$.
\end{itemize}
The behind graph $E^{\mathcal{B}}$ is a directed subgraph of $E$ (see remark \ref{rem:behind_graph}). To define the relaxed behind graph, we identify unnecessary edges. If two agents $i$ and $j$ are neighbors then it is possible that their behind regions overlap and therefore possible that a third agent, $k$, might be in both behind regions.  In terms of the behind graph this means that there is an edge from agent $i$ to agent $k$ \textit{and} from agent $j$ to agent $k$.  However since $i$ and $j$ are neighbors, only one of those edges needs to be present in the behind graph for the interaction graph of the configuration of agents to remain connected; $i$ does not need to take care of $k$ if $j$ is already doing so (or vice versa).  Therefore, we could remove one of those edges from the behind graph creating a new edge set $\widetilde{E}^{\mathcal{B}}$ (see Figure \ref{fig:relaxing_graph}); if the configuration of agents evolves according to the NOLB dynamics in terms of the relaxed edge set $\widetilde{E}^{\mathcal{B}}$ the interaction graph will still remained connected as paths between agents are maintained.

We now define formally a relaxed behind graph.
\begin{defn}
  Given a configuration of agents $\set{{\bf x}_{i}}_{i}\subseteq \Rn$, their corresponding behind graph $G^{\mathcal{B}} = (V, E^{\mathcal{B}})$ and interaction graph $G = (V,E)$, we say that $\widetilde{G}^{\mathcal{B}} = (V, \widetilde{E}^{\mathcal{B}})$ is a \textbf{relaxed behind graph} of $G^{\mathcal{B}}$ if:
  \begin{itemize}
  \item $\widetilde{G}^{\mathcal{B}}$ is a subgraph of $G^{\mathcal{B}}$
  \item for any $(i,j)\in E$, if $(i,k) \in \widetilde{E}^{B}$ then $(j,k) \notin \widetilde{E}^{B}$.
  \end{itemize}
\end{defn}

\begin{figure}[ht]
  \centering
  \includegraphics[scale = 2.2]{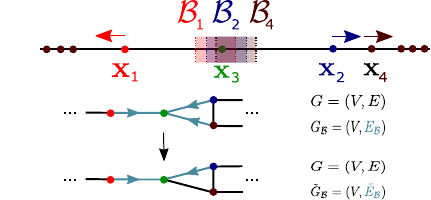}
  \caption{An example of how the behind graph can be relaxed while still ensuring that the interaction graph remains connected.  The interaction graph is represented by undirected and directed edges, the behind graph is represented by only the blue directed edges. Agent 3 is in the behind region of both agent 2 and agent 4 and agents 2 and 4 are connected in the interaction graph therefore we may remove the edge from agent 4 to agent 3.}
  \label{fig:relaxing_graph}
\end{figure}

Note that the relaxed behind graph is not unique as one has degrees of freedom in which edges are removed ($(i,k)$ is removed \textit{or} $(j,k)$).  However, given any full behind graph any two relaxed behind graphs will have the same number of edges.  We can now define a relaxed version of the NOLB dynamics.  Intuitively this new model is exactly the NOLB dynamics however instead of using the full behind graph it is defined in terms of a relaxed behind graph.

\begin{model}[RNOLB]\label{model:relaxed_no_one_left_behind}
  Let $\set{{\bf x}_{1},...,{\bf x}_{N}}\subseteq \Rn$ be a configuration of agents with behind graph $G^{B}$ and let $\widetilde{G}^{\mathcal{B}} = (V, \widetilde{E}^{\mathcal{B}})$ be a relaxed behind graph corresponding to $G^{B}$.  The \textbf{relaxed no one left behind (RNOLB)} dynamics are given by:
  \begin{equation}\label{eq:relaxed_no_one_left_behind}
    {\bf x}_{i}' =  P_{\widetilde{\mathcal{C}}_i}\big(\overline{\bf x}_{i} - {\bf x}_{i}\big)
  \end{equation}
  where $P_{\widetilde{\mathcal{C}}_{i}}: \Rn \rightarrow \widetilde{\mathcal{C}}_{i}$ is the projection operator associated to the cone of velocities $\widetilde{\mathcal{C}}_{i}$ given by:
  \begin{equation}
    \widetilde{\mathcal{C}}_{i} = \{ {\bf v}\in\mathbb{R}^d\,|\;\langle  {\bf v}, {\bf x}_{j}-{\bf x}_{i} \rangle \geq 0\quad \forall  j\;\text{such that}\;(i,j)\in \widetilde{E}^{\mathcal{B}}  \}.
  \end{equation}

\end{model}
We now present an algorithm to easily calculate the relaxed behind graph.  Intuitively, at each time, $t$, an order is randomly computed for the agents.  Then, according to the order each agent projects its velocity towards agents in its behind region that have not already been projected towards by neighboring agents earlier in the order; each agent "takes care of" agents in its behind region that have not already been taken care of by one of its neighbors.  In Figure \ref{fig:disconnect_reconnect} we demonstrate that under these dynamics that individual agents are indeed allowed to disconnect however the connection of the whole configuration is maintained.

\begin{algorithm}[ht]
  \caption{Compute relaxed behind graph}
  \begin{algorithmic}[1]
    \Procedure{Compute relaxed behind graph}{}
      \State Choose an order $\sigma \sim \text{Unif}(\text{permutations of}\; \set{1,\dots,N})$
    \For{$i\in \set{1,\dots,N}$}
      \For{$j\in \set{1,\dots,N}$}
        \If {$(i,j)\in E^{\mathcal{B}}\;\text{and there is no }k\;\text{such that } (i,k)\in E \text{ and } (k,j)\in\widetilde{E}^{\mathcal{B}}$}
        % \If {$(i,j)\in E^{\mathcal{B}}\;\text{and}\;(k,j)\notin\widetilde{E}^{\mathcal{B}}\;\forall k\;\text{such that}\; (i,k)\in E\;\text{and}\;\sigma(k)<\sigma(i)$}
           \State Add $(i,j)$ to $\widetilde{E}^{\mathcal{B}}$
          \EndIf
        \EndFor
      \EndFor
    \EndProcedure
  \end{algorithmic}
\end{algorithm}

%  Let $y\in \Rn$ and $\set{x_{k}}_{k} \subset \Rn$ we define
%   \begin{equation}\label{eq:getting_closer_cone}
%       \mathcal{C}(y, \set{x_{k}}) := \set{v\in\Rn | \langle v, x_{k} - y \rangle \geq 0},
%   \end{equation}
% so $\mathcal{C}(y, \set{x_{k}})$ is the set of directions in which one can move to get closer to all points in $\set{x_{k}}_{k}$ starting from $y$.  We can now define the dynamics.  We also introduce notation for the \textit{set of neighbors} of an agent.  Given a configuration of agents $\set{x_{1},...,x_{N}}$ we define:
% \begin{equation*}
%   N_{i} = \set{x_{j} | |x_{i} - x_{j}| \leq 1}.
% \end{equation*}
% We can now introduce the dynamics.

% \begin{model}[Relaxed no one left behind]
%   Let $\set{x_{1},...,x_{N}}$ be a configuration of agents.  Let $\set{\mathcal{O}_{t}}_{t}$ be a stochastic process where $\mathcal{O}_{t}\thicksim \text{Unif}(\text{permutations of }\set{1,...,N})$ for all $t$.  The \textbf{relaxed no one left behind dynamics (RNOLB)} are given:
%   \begin{equation}\label{eq:RNOLB}
%     \mathcal{P}_{\mathcal{C}(x_{i}, T_{I}^{t})}(\overline{x}_{i} - x_{i})
%   \end{equation}
%   where:
%   \begin{equation}\label{eq:taking_care_of}
%     T_{i}^{t} := \set{x_{j} | x_{j}\in \mathcal{B}_{i}\;\text{and}\; x_{j}\notin T_{k}^{t}\;\text{for all}\; x_{k}\in N_{i}\; \text{such that}\; \mathcal{O}_{t}(k) < \mathcal{O}_{t}(i)}
%   \end{equation}
% \end{model}

Notice that in the example of the NOLB dynamics in Figure \ref{fig:disconnect_reconnect} that the connectivity of the whole configuration of agents is maintained and further that once any two agents connect they remain connected.  In particular this is true for the agents corresponding to the red and blue trajectories.  However in the RNOLB dynamics these two agents become immediately disconnected.  Nonetheless, the connectivity of the whole configuration is maintained as the green agent "takes care" of the blue agent and preserves the existence of a path between the red and blue agents. Indeed, it is clear that two agents $i$ and $j$ connected by a path cannot become disconnected.
\begin{prop}
  The RNOLB dynamics maintain connectivity of the whole configuration of agents.
\end{prop}

\begin{figure}[ht]
  \centering
    \includegraphics[scale = .8]{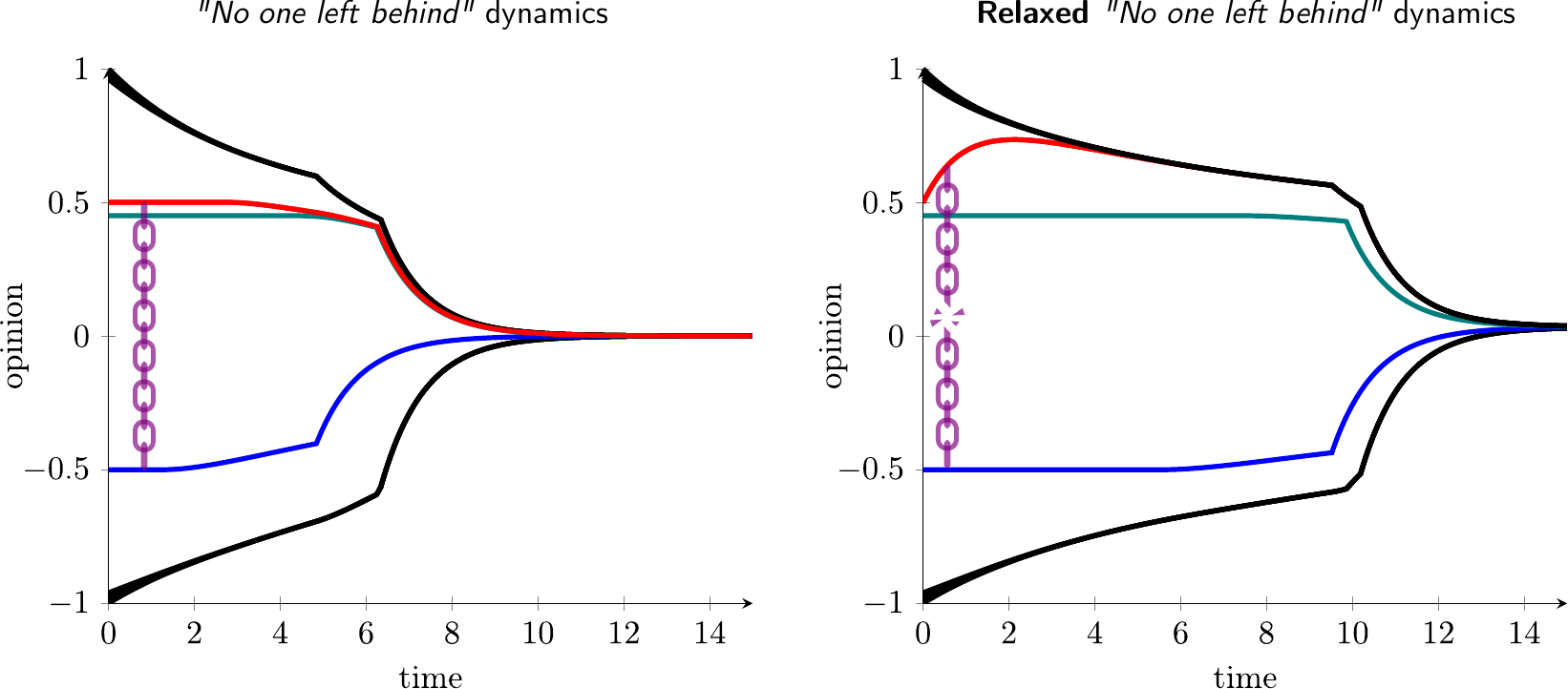}
    \caption{The NOLB dynamics do not allow the red agent to disconnect from the blue agent (illustrated with a purple chain).  The RNOLB dynamics allow this disconnection to occur but maintain connectivity of the whole configuration.}
    \label{fig:disconnect_reconnect}
\end{figure}

  % \begin{prop}
  %   The RNOLB dynamics maintain connectivity of the whole configuration of agents
  % \end{prop}
  % \begin{proof}
  %   Sketch: If $i$ and $j$ are not connected by a path then at some point a path between them had two bad direct neighbors disconnect which implies that no neighbors of the bad neighbors took care of each other, a contradiction to the definition of the dynamics.
  % \end{proof}

  \subsection{RNOLB as an interpolation between NOLB and bounded confidence}
  In this section we conduct a numerical experiment in one dimension to demonstrate that, in a sense, the RNOLB dynamics can be seen as an "interpolation" between the Hegselmann  Krause dynamics \eqref{eq:opinion_formation} and the NOLB dynamics defined in Model \ref{model:no_one_left_behind}.  One of the hallmarks of the bounded confidence dynamics is the formation of clusters of opinions. The requirement of the NOLB dynamics that agents not move if there is another agent in their critical region prevents the formation of clusters (see for example, Figures \ref{fig:example_withWithout_ctrl} and \ref{fig:2D_timelapse}).  Agents with opinions on the interior of the convex hull of opinions always have another agent in their critical region and are prevented from moving until the boundary of the convex hull has contracted sufficiently close to them.  We will see that the weaker conditions of the RNOLB model allow cluster formation to occur initially while still maintaining convergence to a consensus (for a connected initial condition).

  To measure the amount of clustering in a configuration of agents we introduce a simple metric.  Let $R = L/N$ where $L$ is the length of the range of possible opinions and $N$ is the number of agents.  Given a configuration of agents $\set{{\bf x}_{i}}_{i}$ and an agent in the configuration ${\bf x}_{j}$, we count the number of agents who are within $R$ of ${\bf x}_{j}$.  We then take the average over all agents in the configuration.  We will refer to this metric as the \textit{clustering number}.  If the configuration of agents is initially uniformly distributed the clustering number is $2$ (in one dimension).  If agents begin to cluster the clustering number should increase as agents begin to collect more neighbors within $R$.  Clearly, the maximum clustering number for any configuration is the number of agents in the configuration and if a consensus is reached that maximum will be attained.

  In the top three plots of Figure \ref{fig:interpolation_comparison} we show the long term behavior of the RNOLB, NOLB, and bounded confidence dynamics for the same initial condition.  As expected, the bounded confidence dynamics do not reach a consensus and the NOLB and RNOLB dynamics do. Additionally, we qualitatively observe the formation of distinct clusters in the bounded confidence dynamics and the lack of clusters in the NOLB dynamics.  Interestingly, in the RNOLB dynamics we initially observe the formation of clusters that are qualitatively very similar to those observed in the bounded confidence dynamics.  However, instead of remaining distinct (as in the bounded confidence model) these clusters eventually merge and a consensus is reached.
  \begin{figure}[ht]
    \includegraphics[scale = .8]{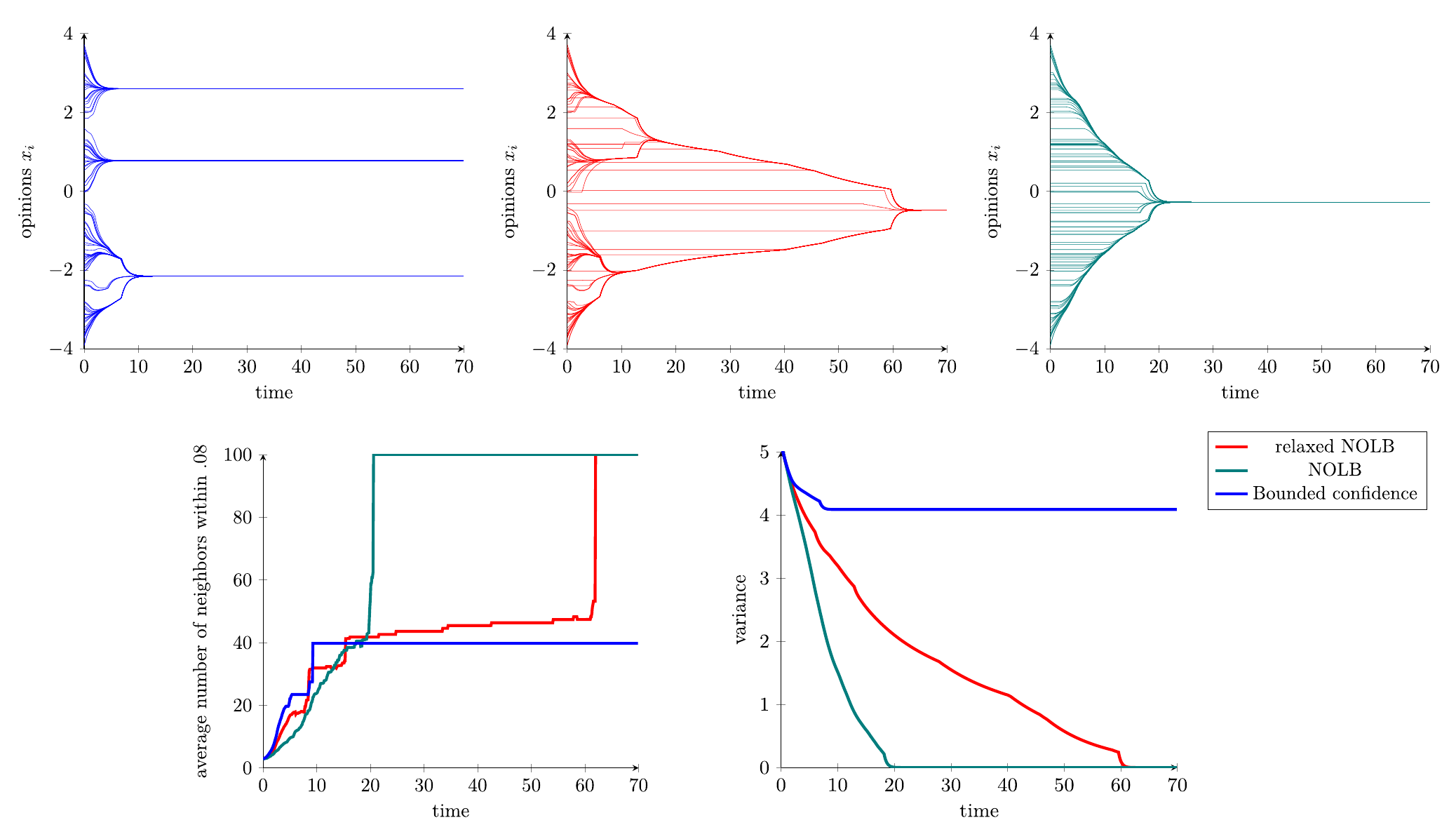}
    \caption{The RNOLB dynamics can be seen as an interpolation between NOLB and bounded confidence.}
    \label{fig:interpolation_comparison}
  \end{figure}

  This qualitative observation is supported by measuring the clustering number of each configuration as time evolves - this is plotted in the bottom left plot of Figure \ref{fig:interpolation_comparison}.  Notice that in the initial period of cluster formation (roughly before $t=10$) that the clustering number of the bounded confidence dynamics is the fastest to increase whereas the NOLB dynamics is the slowest which reflects the strong formation of clusters in the bounded confidence dynamics and weak cluster formation in NOLB. Notably, in this period the clustering number of the RNOLB dynamics increases more slowly than the clustering number of the bounded confidence dynamics but faster than the clustering number of the NOLB dynamics. This supports our qualitative observation that RNOLB is in a sense an interpolation of bounded confidence and NOLB as it has (at least initially) more clustering than NOLB as it allows for free movement of many agents on the interior of the convex hull, but less clustering than bounded confidence as it forces some "moderating" agents to maintain their position in order to preserve connectivity and eventually reach consensus.  We can also observe this interpolation in the evolution of the variance of each configuration which is shown in the bottom right plot of Figure \ref{fig:interpolation_comparison}.  Counter-intuitively, strong clustering causes the decay speed of the variance to reduce as agents in individual clusters can (at least initially) move away from each other.  This effect is observed in the example in Figure \ref{fig:interpolation_comparison} as the NOLB dynamics have the fastest decay in variance, the bounded confidence dynamics have the slowest, and the variance decay of RNOLB is faster than bounded confidence but slower than NOLB which reflects the varying amount of clustering observed in the three different models.  We also observe that due to the higher degree of clustering, the RNOLB dynamics are slower to converge to a consensus then the NOLB dynamics.  Despite this, we still observe exponential convergence once all agents are within the interaction range of each other.

\subsection{Diameter decay in RNOLB}

In our previous results concerning the NOLB dynamics we show rigorously that convergence to a consensus occurs in two stages; when the diameter of the configuration of agents is greater than one the rate is linear, afterwards it spontaneously becomes exponential.  Our estimates of these rates were fairly rough and in an attempt to discover whether the rates are faster in practice we simulated 100 realizations of the NOLB dynamics and found a large disparity in stopping times which suggests that our estimates are unlikely to be improved.  Here, we repeat this experiment for the RNOLB dynamics in order to investigate whether there is a similar phase transition in the convergence and disparity in stopping times.

We perform $100$ realizations of the RNOLB dynamics with initial conditions drawn from a uniform distribution on the interval [0, 10].  We then again compute the evolution of the diameter of the configuration $d(t)$ for all realizations and plot the decay of the median of the diameter, the slowest and fastest decay, and the 5\% and 95\% quantiles.  Here, we again observe two phases of convergence with clear exponential convergence again emerging once the diameter reaches $1$.  However, compared to the NOLB dynamics convergence is much slower and there is a much larger disparity between stopping times.  The fastest evolution reaches a diameter of 1 at $\approx 50$ time units whereas the slowest evolution takes greater than $200$ time units.  This suggests that it will be difficult to obtain tight convergence rates in the case of the RNOLB dynamics as well.

\begin{figure}[ht]
  \centering
  \includegraphics[scale=1]{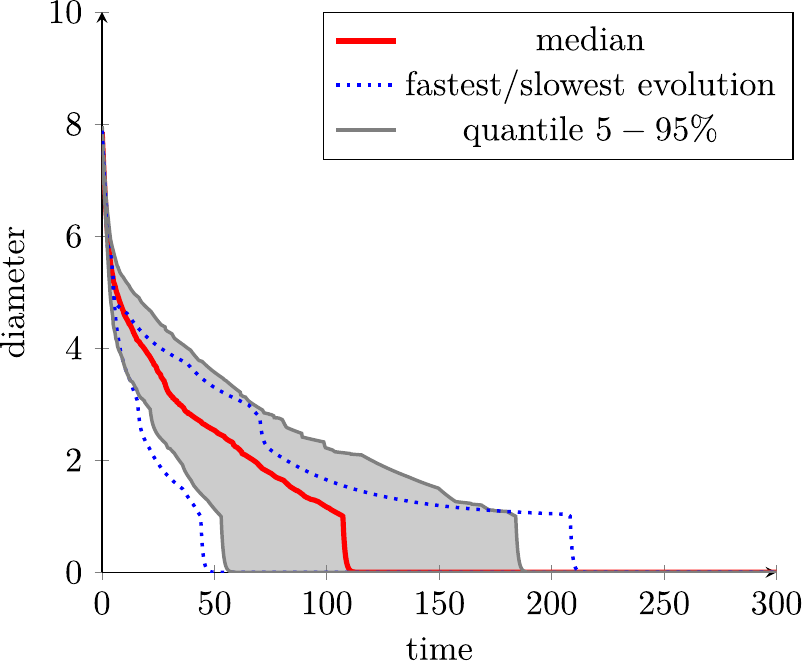}
  \caption{Diameter, $d(t)$ over time for 100 realizations (quantile representation).}
\end{figure}

\section{Conclusion and future work}

In this paper we studied variants of the Hegselmann -Krause bounded confidence dynamics introduced in \cite{hegselmann_opinion_2002}.  The modifications we introduced were aimed at mitigating the generic cluster forming behavior seen in the bounded confidence dynamics and inducing consensus among the agents.  Motivated by the attractive nature of the interaction in bounded confidence dynamics we introduced a variant dubbed \textit{No one left behind} (NOLB) that maintains connectivity between agents.  We rigorously demonstrated that this control is sufficient for unconditional convergence to a consensus regardless of the dimension of agent opinions.  Due to the nonlinear and discontinuous nature of the dynamics the argument relies on the interplay between two key properties of the dynamics; contractivity and the preservation of connectivity of the configuration of agents.  In one dimension we were able to derive explicit convergence rates that quantify how fast a consensus is reached.  Additionally, we conducted numerical experiments that suggest that tighter bounds on the convergence rates are likely not possible.

The NOLB dynamics maintain local, pairwise connectivity between agents, however our argument for unconditional convergence to a consensus relied only on global connectivity of the configuration of agents.  Motivated by this we introduced a second variant to the bounded confidence dynamics we dubbed \textit{Relaxed no one left behind} (RNOLB) aimed at maintaining the existence of paths between pairs of agents.  We find that while this modification still results in unconditional convergence to a consensus, it retains more of the qualitative features of the bounded confidence dynamics; notably the emergence of clusters in the beginning of its evolution.  For this reason the RNOLB dynamics can be regarded as an interpolation between the bounded confidence dynamics and the NOLB dynamics.  We presented numerical investigations into the variance of agent opinions and a metric we dub the \textit{clustering number} that further support this view.

To our knowledge, models with a bounded confidence type interaction (one that depends on an interaction radius) have only been studied in a Euclidean setting.  Motivated by this observation we hope to study the behavior of models with a bounded confidence style interaction in topology other than Euclidean space, eg. the circle.  It would be interesting to extended this study and investigate whether the corresponding NOLB and RNOLB dynamics result in consensus in such spaces.
Empirical studies of models of decentralized collective behavior outside of statistical physics have been challenging in the past due to lack of data that both describes their motion and their pairwise interactions.  However, there has been some progress in recent years in biological studies of swarming animals such as birds and ants \cite{ballerini_empirical_2008, buhl_disorder_2006, poissonnier_experimental_nodate}.  We believe the advent of social media provides a data source through which an empirical study of consensus in the context of opinion formation could be possible \cite{estrada_communicability_2015, kempe_maximizing_2003}.  A possible goal of such a study would be to confirm the phenomenon found in this study; connectivity in the interaction network of agents is a main driver of consensus.

\appendix
\section{Appendix}

\setcounter{theorem}{1}
\setcounter{lem}{0}

\begin{lem}
  \label{lem:project_help_1}
  Let $\set{u_{i}}_{i}\subseteq\Rn$ and consider $C = \set{v | \inner{v}{u_{i}}\geq 0\;\text{for all $i$}}$.  If $u$ is in $C$ then $\inner{P_{C}(x)}{u}\geq\inner{x}{u}$ for all $x$ in $\Rn$.
\end{lem}
\begin{proof}
  $P_{C}(x)$ is by definition the solution to the minimization problem:
  \begin{align*}
    &\text{minimize}\;f(y) = \frac{1}{2}\|y - x\|^{2} \\
    &\text{subject to}\; g_{i}(y) = \inner{y}{u_{i}}\geq 0
  \end{align*}
  Since $C$ is closed and convex there exists a unique minimizer $y^{*}$ and since the constraints are linear the Karush-Kuhn-Tucker conditions imply that:
  \begin{align*}
    \nabla f(y^{*}) = \sum_{i} \lambda_{i}\nabla g_{i}(y^{*})
  \end{align*}
  where $\lambda_{i}\geq 0$.  Therefore:
  \begin{align*}
    y^{*} = x + \sum_{i}\lambda_{i}u_{i}
  \end{align*}
  which implies by definition of $C$ that:
  \begin{align*}
    \inner{P_{C}(x)}{u} &= \inner{x + \sum_{i}\lambda_{i}u_{i}}{u} = \inner{x}{u} + \sum_{i}\lambda_{i}\inner{u_{i}}{u} \geq \inner{x}{u}
  \end{align*}
  as desired.
\end{proof}

\begin{lem}
  \label{lem:delta_T}
  For any $0<r_*<1$, $\eta>0$ and $N>0$, there exist $T>0$ and $\delta>0$ such that:
  \begin{eqnarray}
    \delta+2T &\leq& \min(r_*,1-r_*) \\               % |x_p2-x_p1|\geq1-r_*, |x_p-x_p1|\leq1-r_*
    \frac{T^2}{2N^2} \eta\big(\eta (1-\delta-2T)-2N(\delta+2T)\big) &\geq& \delta. % x_p connects with x_p2
  \end{eqnarray}
                                                        %                                                         {\color[rgb]{0,.5,0} It seems that there is no need for the '2' (see proof theorem):\\
      %       $\frac{T^2}{2N^2} \eta\big(\eta (1-\delta-2T)-2N(\delta+{\bf 1}T)\big) \geq \delta$
      %       }
\end{lem}
\begin{proof}

      %       We need to find a solution $(\delta_{*},T_{*})$ such that above inequality holds.
  We let $\delta=T^3$ and show that for $T>0$ sufficiently small both equations are satisfied. Indeed, the substitution leads to:
  \begin{eqnarray}
    \label{eq:delta_T_eq1_bis}
    T^3+2T &\leq& \min(r_*,1-r_*) \\
    \label{eq:delta_T_eq2_bis}
    \frac{1}{2N^2} \eta\big(\eta (1-T^3-2T)-2N(T^3+2T)\big) &\geq& T
  \end{eqnarray}
  Since $\min(r_*,1-r_*)>0$ and $T^3+2T \stackrel{T \to 0}{\longrightarrow} 0$, there exists $T_1>0$ such that \eqref{eq:delta_T_eq1_bis} is satisfied for $0<T\leq T_1$. Similarly, for the equation \eqref{eq:delta_T_eq2_bis}, we notice that
  \begin{displaymath}
    \frac{1}{2N^2} \eta\big(\eta (1-T^3-2T)-2N(T^3+2T)\big) \;\;\stackrel{T \to 0}{\longrightarrow}\;\; \frac{\eta^2}{2N^2} \;\;> 0.
  \end{displaymath}
  Therefore, there exists $T_2>0$ such that \eqref{eq:delta_T_eq2_bis} is satisfied for $0<T\leq T_2$. Taking $T=\min(T_1,T_2)$ and $\delta=T^3$, we deduce a solution to \eqref{eq:delta_T_eq1}-\eqref{eq:delta_T_eq2}.
\end{proof}

\bibliography{bib}
\bibliographystyle{plain}

\end{document}